\documentclass[11pt]{article}

\usepackage{fullpage}

\usepackage{dsfont}
\usepackage{latexsym}
\usepackage{amsmath,amssymb,amsfonts,amsthm}
\usepackage{mathtools,mathrsfs}
\usepackage{dsfont}

\usepackage{aliascnt}
\usepackage{graphicx,color}
\usepackage[dvipsnames]{xcolor}
\usepackage{multirow,enumerate}
\usepackage{cite,url}
\usepackage{float,xfrac}

\usepackage{algorithmic}
\usepackage{algorithm}

\usepackage{slashbox}

\newtheorem{theorem}{Theorem}%

\newaliascnt{lemma}{theorem}
\aliascntresetthe{lemma}

\newaliascnt{claim}{theorem}
\newtheorem{claim}[claim]{Claim}%
\aliascntresetthe{claim}

\newaliascnt{corollary}{theorem}
\aliascntresetthe{corollary}

\newaliascnt{proposition}{theorem}
\newtheorem{proposition}[proposition]{Proposition}%
\aliascntresetthe{proposition}

\newaliascnt{remark}{theorem}

\aliascntresetthe{remark}

\newaliascnt{result}{theorem}

\aliascntresetthe{result}

\newaliascnt{algo}{procedure}
\aliascntresetthe{algo}

\newtheorem{definition}{Definition}

\newcommand{\AutoAdjust}[3]{\mathchoice{ \left #1 #2  \right #3}{#1 #2 #3}{#1 #2 #3}{#1 #2 #3} }
\newcommand{\Xcomment}[1]{{}}

\newcommand{\InBrackets}[1]{\AutoAdjust{[}{#1}{]}}
\newcommand{\Ex}[2][]{\operatorname{\mathbf E}_{#1}\InBrackets{#2}}

\newcommand{\Prx}[2][]{\operatorname{\mathbf{Pr}}_{#1}\InBrackets{#2}}
\def\prob{\Prx}
\def\expect{\Ex}

\newcommand{\vect}[1]{\ensuremath{\mathbf{#1}}}

\newcommand{\R}{\mathbb{R}}

\newcommand{\diff}{{\,\mathrm{d}}}

\newcommand{\bx}{{\vect{x}}}

\newcommand{\bz}{{\vect{z}}}
\newcommand{\bzmi}[1][i]{{\bz_{\text{-}#1}}}
\newcommand{\bbarz}{{\overline{\vect{z}}}}
\newcommand{\bhatz}{{\hat{\vect{z}}}}

\newcommand{\ind}{{\mathds{1}}}

\newcommand{\auction}{\mathcal{A}}
\newcommand{\nbidder}{n}

\newcommand{\val}{v}
\newcommand{\vali}[1][i]{\val_{#1}}
\newcommand{\vals}{\vect{\val}}
\newcommand{\valsmi}[1][i]{\vals_{\text{-}#1}}

\newcommand{\support}{\Lambda}

\newcommand{\F}{{\mathcal{F}}}
\newcommand{\M}{{\mathcal{M}}}

\newcommand{\proj}{{_{\downarrow}}}
\newcommand{\proji}[1][i]{{\proj_{#1}}}

\newcommand{\weight}{w}

\newcommand{\tight}{\Gamma}
\newcommand{\bn}{\mathbf{B}}

\newcommand{\price}{p}

\newcommand{\pricei}[1][i]{{\price_{#1}}}

\newcommand{\bid}{b}
\newcommand{\bidi}[1][i]{{\bid_{#1}}}
\newcommand{\bids}{\vect{\bid}}
\newcommand{\bidsmi}[1][i]{{\bids_{\text{-}#1}}}

\newcommand{\hbid}{h}

\newcommand{\hbids}{\vect{\hbid}}
\newcommand{\hbidsmi}[1][i]{{\hbids_{\text{-}#1}}}

\newcommand{\wbid}{d}

\newcommand{\wbids}{\vect{\wbid}}
\newcommand{\wbidsmi}[1][i]{{\wbids_{\text{-}#1}}}

\newcommand\maxV{\mbox{{\sc maxV}}}
\newcommand\LS{\mbox{{\sc LS}}}

\title{Optimal Competitive Auctions}

\author{
Ning Chen\thanks{Nanyang Technological University, Singapore. Email: \texttt{ningc@ntu.edu.sg}.}
\and
Nick Gravin\thanks{Microsoft Research. Email: \texttt{ngravin@microsoft.com}.}
\and
Pinyan Lu\thanks{Microsoft Research. Email: \texttt{pinyanl@microsoft.com}.
}
}

\date{}

\begin{document}

\maketitle

\begin{abstract}

We study the design of truthful auctions for selling identical items in unlimited supply (e.g., digital goods) to $\nbidder$ unit demand buyers. This classic problem stands out from profit-maximizing auction design literature as it requires no probabilistic assumptions on buyers' valuations and employs the framework of competitive analysis. Our objective is to optimize the worst-case performance of an auction, measured by the ratio between a given benchmark and revenue generated by the auction.

We establish a sufficient and necessary condition that characterizes competitive ratios for {\em all} monotone benchmarks.
The characterization identifies the worst-case distribution of instances and reveals intrinsic relations between competitive ratios and benchmarks in the competitive analysis. With the characterization at hand, we show optimal competitive auctions for two natural benchmarks.

The most well-studied benchmark $\F^{(2)}(\cdot)$ measures the envy-free optimal revenue where at least two buyers win. Goldberg et al.~\cite{GoldbergHKS04} showed a sequence of lower bounds on the competitive ratio for each number of buyers $\nbidder$. They conjectured that all these bounds are tight. We show that optimal competitive auctions match these bounds. Thus, we confirm the conjecture and settle a central open problem in the design of digital goods auctions. As one more application we examine another economically meaningful benchmark, which measures the optimal revenue across all limited-supply Vickrey auctions. We identify the optimal competitive ratios to be $(\frac{\nbidder}{\nbidder-1})^{\nbidder-1}-1$ for each number of buyers $\nbidder$, that is $e-1$ as $\nbidder$ approaches infinity.

\end{abstract}

\setcounter{page}{0}\thispagestyle{empty}
\newpage


\section{Introduction}\label{sec:introduction}

A central question in auction theory is to design optimal truthful auctions that maximize
the revenue of the auctioneer. A truthful auction must be incentive compatible with the selfish
behavior of bidders and encourage them to reveal their private information truthfully.
In the classic single-parameter setting, the auctioneer sells multiple copies of an item to $n$
unit-demand bidders, each bidder having a private value $\vali$ for the item.
The revenue maximization problem in economics is traditionally analyzed in the Bayesian framework.
The seminal Myerson's auction~\cite{Myerson81} provides an optimal design that extracts the maximum expected revenue for a given distribution of values.

While Myerson's auction gives an optimal design in the Bayesian framework, in many scenarios determining or estimating the prior distribution in advance is impossible. Without prior information on the distribution, how should an auction be designed? The auctioneer not being certain about the prior distribution is likely to resort to truthful auctions that generate a good revenue on every possible instance of bidders' valuations. We therefore employ the worst-case  competitive analysis, which characterizes the worst-case performance of an auction across all possible bidders' valuations. We follow the framework summarized in \cite{GoldbergHKSW06} where a truthful auction is required to generate a profit comparable to the value of a certain economically meaningful benchmark on every input of values. The motivation of this framework comes from the analysis of on-line algorithms, where the performance of an on-line algorithm, which is unaware of the future, is measured in terms of the performance of the optimal off-line algorithm that knows the future. The assumption that an on-line algorithm does not know the future in advance corresponds to the assumption that an auction does not know the bidders’ valuations in advance. How well an auction can perform in the worst-case? We will answer this question in the present paper.

A particularly interesting single-parameter setting is {\em digital goods auctions} in which the number of units for sale is unlimited or is greater than or equal to the number of bidders. Digital goods auctions are motivated by the applications of selling digital goods like downloadable software on the Internet or pay-per-view television where there is a negligible cost for producing a copy of the item. The design of the optimal Bayesian auction for digital goods is
trivial: given a prior distribution of values, one may independently offer each bidder a fixed price $p_i$ tailored to the value's distribution $\vali$ so as to maximize the expected revenue $p_i\cdot \mathbf{Pr}[v_i\ge p_i]$. From the worst-case perspective, digital goods auctions attracted considerable attention over the last decade \cite{GoldbergHW01,FiatGHK02,GoldbergH03,GoldbergHKS04,HartlineM05,FeigeFHK05,AlaeiMS09,IchibaI10}. However, the optimal design was still unknown. The question about designing optimal digital goods auctions remained widely open and is fundamental to our understanding of the optimal competitive auction design in the general single-parameter setting.

Formally, in the competitive analysis framework, 
given a benchmark function $f(\cdot)$, the {\em competitive ratio} (with respect to $f(\cdot)$) of a truthful auction $\auction$ is defined as $\max_{\vals} \frac{f(\vals)}{\Ex{\auction(\vals)}}$, where $\Ex{\auction(\vals)}$ is the expected revenue of the auction $\mathcal{A}$ on valuation vector $\vals=(v_1,\ldots,v_n)$. The objective is to design an auction that minimizes the competitive ratio with respect to a given benchmark.

We note that a number of functions can serve as meaningful benchmarks and that they may have different optimal competitive ratios. Such flexibility in the choice of a target benchmark function may prove to be helpful for modeling different objectives of the auctioneer. Conversely, it may allow us to make meaningful conclusions and prescriptions about target benchmarks. However, most of the work done along the line of competitive analysis in algorithmic game theory~\cite{AGT-book} is devoted to the competitive analysis of a well-motivated but fixed concrete benchmark function. Namely, while there are basic guidelines for choosing a good benchmark function (e.g., the benchmark should have a strong economic motivation and should match the performance of the optimal auction as closely as possible), there are no formal criteria for distinguishing different benchmarks. In this paper, instead of justifying what a good benchmark is, we focus on the intrinsic relation between benchmarks and their corresponding competitive ratios. We give a complete characterization for almost {\em all} possible benchmark functions.

\medskip
\noindent \textbf{Theorem 1.} (Characterization) For any non-negative and monotonically increasing function $f(\cdot)$, there is a truthful digital goods auction that achieves a competitive ratio of $\lambda$ with respect to $f(\cdot)$ if and only if
\begin{equation}\label{eq-intro-char}
\int\limits_{S}f(\vals)\cdot\weight(\vals)\diff\vals\le \lambda\cdot\sum_{i=1}^n \ \int\limits_{\ S\proji}\weight(\valsmi)\diff\valsmi,
\end{equation}
where $S$ is any upward closed set in the support of valuation vectors, $S\proji$ is the projection of $S$ along the $i$-th coordinate, $\weight(\vals)= \prod_{i} 1 / \vali^2$, and $\valsmi = (\vali[1],\ldots,\vali[i-1],\vali[i+1],\ldots,\vali[n])$.
\medskip

The fact that the weight function $\weight(\vals)= \prod_{i} 1 / v_i^2$ appears in the above theorem is not a coincidence. The corresponding single-parameter distribution with the density function $\weight(\vali) = 1/\vali^2$ for $\vali\in[1,\infty)$ is a common tool to provide bounds on the performance of auctions. It is called {\em equal-revenue distribution}, as it enjoys a remarkable property that if the value of a bidder $i$ is drawn from this distribution, then any fixed price $p_i$ offered to $i$ generates the same expected revenue $p_i\cdot \mathbf{Pr}[v_i\ge p_i]=1$. In particular, the product $\weight(\vals)$ of independent equal-revenue distributions $\weight(\vali)$ was used in~\cite{GoldbergHKS04} to obtain the best known lower bounds on the competitive ratios of digital goods auctions; it has also been used in the auction analysis of other models, e.g., in~\cite{HN12}.

As we know from Yao's minimax principle, for any benchmark function $f(\cdot)$, there is a distribution of instances such that on average for that distribution no truthful auction can beat the worst-case competitive ratio with respect to $f(\cdot)$. The inequality \eqref{eq-intro-char} in the theorem states that the equal revenue distribution is indeed the worst-case distribution to estimate the competitive ratio of an auction.
In other words, in the context of digital goods auctions, the worst-case distribution can be described only by its support, and the actual density of the distribution is given by the equal-revenue distribution. The theorem gives a sufficient and necessary condition for a benchmark function $f(\cdot)$ to admit a competitive ratio of $\lambda$. It implies that the optimal competitive ratio with respect to $f(\cdot)$ is the smallest value of $\lambda$ for which the inequality \eqref{eq-intro-char} holds. The theorem indicates that all benchmarks and their competitive ratios are tied to the equal revenue distribution: it is the worst distribution not only for the benchmark considered in \cite{GoldbergHKS04} but also for {\em all} monotone benchmarks.


We note that our characterization is provided by the set of inequalities \eqref{eq-intro-char} that only involves a function $f(\cdot)$ and a ratio $\lambda$, but does not describe an actual auction. This is similar in spirit to the characterization of truthfully implementable allocation functions, where an allocation function is {\em truthfully implementable} if there exists a payment function that makes the allocation function truthful. To characterize a truthfully implementable allocation function, one uses a set of inequalities such as weak-monotonicity \cite{SY05,BikchandaniCLMNS06,ABHM10} that only specifies an allocation function but says nothing about payments. Hence, without describing any payment function, one can determine whether there is a truthful auction with a specified allocation. Our characterization of benchmarks with a given competitive ratio shares a similar philosophy: the condition determines whether there exists a truthful auction with a certain completive ratio with respect to $f(\cdot)$, but does not explicitly describe the auction.


The characterization theorem provides us with a powerful tool to analyze the optimal competitive ratios of auctions for different benchmarks.
We first consider the most well-studied benchmark $\F^{(2)}$ introduced in~\cite{GoldbergHW01}, which is defined as the maximum revenue achieved in an envy-free allocation provided that at least two bidders receive the item.
We next study another natural benchmark, denoted by $\maxV$, which is defined as the maximal revenue of the $k$-item Vickrey auction across all possible values of $k$. We have the following results on the optimal competitive ratios for these two benchmarks.



\medskip
\noindent \textbf{Theorem 2.} There are truthful digital goods auctions that achieve the optimal competitive ratios of $\lambda_n$ and $\gamma_n$ for any $n\ge2$ with respect to the benchmarks $\F^{(2)}$ and $\maxV$, respectively, where
\[\lambda_\nbidder = 1-\sum_{i=2}^\nbidder \left(\frac{-1}{\nbidder}\right)^{i-1} \frac{i}{i-1} {\nbidder-1 \choose i-1} \mbox{ \ \ and \ \ }
\gamma_n = \left(\frac{n}{n-1}\right)^{n-1}-1.\]

It was shown in \cite{GoldbergHKS04} that $\lambda_\nbidder$ gives a lower bound on the competitive ratio of any auction with respect to $\F^{(2)}$ for any $\nbidder\ge 2$. These lower bounds were obtained by calculating the expected value of the benchmark $\F^{(2)}(\vals)$ when $\vals$ is drawn from the equal revenue distribution. In particular, $\lambda_2=2$, $\lambda_3=13/6$, and in general, $\{\lambda_\nbidder\}_{n=1}^{\infty}$ is an increasing sequence with a limit of roughly $2.42$. Goldberg et al.~\cite{GoldbergHKS04} conjectured that the lower bounds given by $\{\lambda_n\}$ are tight. Indeed, for $\nbidder=2$ bidders, the second price auction gives a matching competitive ratio of $\lambda_2$; for $\nbidder=3$ bidders, \cite{HartlineM05} gave a sophisticated auction with a competitive ratio that matches the lower bound of $\lambda_3$. These are the only cases for which optimal competitive auctions were known.
Our result confirms the conjecture of \cite{GoldbergHKS04} and settles the long standing open problem of designing optimal digital goods auctions with respect to the benchmark $\F^{(2)}$.

For the benchmark $\maxV$, we calculate the expected value of $\maxV(\vals)$ when $\vals$ is drawn from the equal revenue distribution, and derive a sequence of optimal competitive ratios $\gamma_\nbidder$ for each number $\nbidder$ of bidders in the auction. We note that $\gamma_2=1$, $\gamma_3=5/4$, and $\{\gamma_\nbidder\}_{n=1}^{\infty}$ is an increasing sequence with the limit of $e-1$ as $\nbidder$ approaches infinity.



\medskip

Finally, as another application of the characterization theorem, we consider optimal competitive auctions for the multi-unit limited supply setting where there is an item with $k$ units of supply for sale to $\nbidder$ unit-demand bidders. It was observed in \cite{GoldbergHKSW06} that there is a competitive ratio preserving reduction from unlimited supply to limited supply. Namely, given a digital goods auction with a competitive ratio $\lambda$ with respect to $\F^{(2)}$ for $k$ bidders, one can construct a truthful auction for $\nbidder$ bidders with the same competitive ratio $\lambda$ with respect to $\F^{(2,k)}$ (where $\F^{(2,k)}$ is the optimal fixed price revenue provided that at least $2$ and at most $k$ bidders receive the item). Our analysis continues to hold for those benchmarks that only depend on the $k$ highest values, and thus, gives optimal competitive auctions in the competitive analysis framework.

%

\subsection{Related Work}

The study of competitive digital goods auctions was coined by Goldberg et al.~\cite{GoldbergHW01}, where the authors introduced the random sampling optimal price auction and showed that it has a constant competitive ratio with respect to $\F^{(2)}$. Later on the competitive ratio of the auction was shown to be 15 and 4.68, by Feige et al.~\cite{FeigeFHK05} and Alaei et al.~\cite{AlaeiMS09}, respectively. Since the pioneer work of \cite{GoldbergHW01}, a sequence of work has been devoted to design auctions with improved competitive ratios: the random sampling cost sharing auction \cite{FiatGHK02} with ratio 4; the consensus revenue estimate auction \cite{GoldbergH03} with ratio 3.39; the aggregation auction \cite{HartlineM05} with ratio 3.25; the best known ratio 3.12 is attained by the averaging auction \cite{IchibaI10}.


Fiat et al.~\cite{FiatGHK02} formulated the prior-free analysis framework for digital goods auctions. The framework was further developed to general symmetric auction problems and connected with the Bayesian framework in \cite{HartlineR08}. The relation between envy-freedom and prior-free mechanism design was further investigated in \cite{HartlineY11}.
Aggarwal et al. \cite{AggarwalFGHIS05} showed that every randomized auction in the digital goods environment can be derandomized in polynomial time with an extra additive error that depends on the maximal range of values.

Leonardi and Roughgarden~\cite{LR2012} introduced another benchmark, namely monotone-price benchmark $\M^{(2)}$. Later it was shown in~\cite{BKK2013} that digital goods auctions have a constant competitive ratio with respect to $\M^{(2)}$.

A number of variants of digital goods auctions have been investigated, including, e.g., online auctions \cite{BaryossefHW02,BlumH05}, limited supply ($k$-unit auctions) \cite{DevanurH09}, online auctions with unknown limited supply \cite{MahdianS06}, externalities between bidders \cite{GravinL13}, and matroid permutations and position environments \cite{HartlineY11}.



\section{Preliminaries}\label{sec:preliminary}

In a digital goods auction, an auctioneer sells multiple copies of an item in unlimited supply to $\nbidder$ bidders. Each bidder $i$ is interested in a single unit of the item and values it at a privately known value $\vali$. We consider a single-round auction, where each bidder submits a sealed bid $\bidi$ to the auctioneer. Upon receiving submitted bids $\bids=(\bid_1,\ldots,\bid_\nbidder)$ from all bidders, the auctioneer decides on whether each bidder $i$ receives an item and the amount that $i$ pays. If bidder $i$ wins an item, his {\em utility} is the difference between his value $\vali$ and his payment; otherwise, the bidder pays $0$ and his utility is $0$. The auctioneer's {\em revenue} is the total payment of the bidders.

We assume that all bidders are self-motivated and aim to maximize their own utility. We say that an auction is {\em truthful} or incentive compatible if it is a dominant strategy for every bidder $i$ to submit his private value, i.e., $\bidi=\vali$, no matter how other bidders behave. A randomized auction is (universally) truthful if it is randomly distributed over deterministic truthful auctions.

An auction is called {\em bid-independent} if, for each bidder $i$, the auctioneer computes a threshold price $\pricei$ according to the bids of the rest $\nbidder-1$ bidders $\bidsmi=(\bid_1,\ldots,\bid_{i-1},?,\bid_{i+1},\ldots,\bid_\nbidder)$. In other words, there is a function $g_i$ such that $\pricei=g_i(\bidsmi)$.
It was shown in~\cite{GoldbergHW01} that an auction is truthful if and only if it is bid-independent. Thus, it is sufficient to consider bid-independent auctions in order to design truthful auctions.

To evaluate the performance of an auction, we need to have a reasonable benchmark function $f: \R^\nbidder\rightarrow\R$,
where $f(\bids)$ measures our target revenue for the bid vector $\bids$.
Given a benchmark function $f(\cdot)$, we say that an auction $\auction$ has a {\em competitive ratio} of $\lambda$ with respect to $f(\cdot)$ if
$$\frac{f(\vals)}{\Ex{\auction(\vals)}}\le \lambda, \quad \forall \vals=(\vali[1],\ldots,\vali[\nbidder])$$
where $\Ex{\auction(\bids)}$ is the expected revenue of auction $\auction$ on the bid vector $\bids$.
The focus of our paper is to design truthful auctions that minimize the competitive ratio with respect to different benchmarks.

In this paper, we assume that a benchmark function $f(\cdot)$ is non-negative and monotone.
These are natural conditions for a function to serve as a reasonable benchmark.
Specifically, we will focus on the following benchmark functions. (Given a bid vector $\bids=(\bid_1,\ldots,\bid_n)$ we reorder bids so that $\bidi[(1)]\ge \bidi[(2)]\ge \cdots \ge \bidi[(\nbidder)]$.)
\begin{itemize}
\item $\F^{(2)}(\bids)= \max_{2\le k\le\nbidder} \ k\cdot \bidi[(k)]$. That is, $\F^{(2)}$ gives the largest possible revenue obtained in a fixed price auction given that there are at least two winners. $\F^{(2)}$ was denoted sometimes as $\mathcal{F}^{(2)}$ in the previous literature and provides the optimal envy-free revenue conditioned on that at least two bidders receive the item.
\item $\maxV(\bids)= \max_{1\le k<\nbidder} \ k\cdot \bidi[(k+1)]$. We note that $k\cdot \bidi[(k+1)]$ is the revenue of the $k$-item Vickrey auction  with a fixed supply of $k$ items. Hence, $\maxV$ gives the largest revenue obtained in the Vickrey auction for selling $k$ items for all possible values of the limited supply $k$.
\end{itemize}

\begin{theorem}[Goldberg et al.~\cite{GoldbergHKS04}]
\label{th:GoldbergHKS_ratio}
The competitive ratio with respect to $\F^{(2)}$ of any truthful randomized auction is at least
$$\lambda_\nbidder = 1-\sum_{i=2}^\nbidder \left(\frac{-1}{\nbidder}\right)^{i-1} \frac{i}{i-1} {\nbidder-1 \choose i-1}.$$
\end{theorem}

We note that $\lambda_2=2$, $\lambda_3=13/6$, and $\{\lambda_\nbidder\}_{n=1}^{\infty}$ is an increasing sequence with a limit of $2.42$ when $n$ approaches infinity.
In the proof of the above theorem, the authors of \cite{GoldbergHKS04} constructed a so-called equal revenue distribution where all values $\vali$ are drawn identically and independently with probability $\Prx{\vali > x} = \frac{1}{x}$ for any $x\ge 1$. A remarkable property of this distribution is that any truthful auction has the same expected revenue $\nbidder$. It was shown that the expected value of $\F^{(2)}(\vals)$ is $\nbidder\times\lambda_\nbidder$. Thus, the theorem follows since
\[
\max_{\vals}\frac{\F^{(2)}(\vals)}{\Ex{\auction(\vals)}}\ge
\frac{\Ex{\F^{(2)}(\vals)}}{\Ex{\auction(\vals)}}=
\frac{\nbidder\times\lambda_\nbidder}{\nbidder} = \lambda_\nbidder.
\]



\section{Characterization of Benchmarks}\label{sec:characterization}

We introduce the following central definition for a Benchmark function.

\begin{definition}[Attainability]
A benchmark function $f(\cdot)$ is {\em $\lambda$-attainable}
if there exists a truthful auction that has a competitive ratio of $\lambda$ with respect to $f(\cdot)$.
\end{definition}

We shall give a sufficient and necessary condition of attainability of a benchmark in this section with a set of inequalities which only involves the function $f$ and ratio $\lambda$ but not any auction.  After having the characterization, we analyze the attainability of two well-studied benchmarks in the next section.

%

For technical simplicity, we consider a discrete and bounded domain for all bids: for any bidder $i$,
we assume that $b_i\in \big\{(1+\delta)^t ~|~ t=0,1,2,\ldots, N\big\}$, where $\delta>0$ is any fixed small constant. Thus, all bids are between 1 and $(1+\delta)^N$, and are multiples of $1+\delta$. Let
$$\support = \big\{(1+\delta)^t ~|~ t=0,1,2,\ldots, N\big\}$$ denote the support of a single bidder's bids and $\support^n$ denote the support of bid vectors of all $n$ bidders.
We note that such a multiplicative discretization is not critical for our characterization. (Alternatively, we may consider an additive discretization with
bids being integer multiples of $\delta$.) We will discuss how to generalize our analysis to continuous and unbounded domains at the end of this section.


For the domain $\support^n$ we assume without loss of generality that for any $\bidsmi$, the price offered to bidder $i$ is also from $\support$, which is of the form $(1+\delta)^t$. For $\bidsmi\in \support^{n-1}$ and $\pricei\in \support$, let $z_i(\bidsmi,\pricei)$ denote the probability that the auctioneer offers price $\pricei$ to bidder $i$ when observing others' bids $\bidsmi$. There exists a truthful auction that is $\lambda$-competitive with respect to the benchmark $f(\cdot)$ if and only if the following linear system is feasible:
\begin{equation*}
\LS_1: \
\begin{cases}
   \lambda \cdot \sum\limits_{i}\sum\limits_{\pricei=1}^{\bidi} \pricei\cdot z_i(\bidsmi,\pricei) \geq  f(\bids), & \forall \bids \\[.2in]
   \sum\limits_{\pricei=1}^{(1+\delta)^N}z_i(\bidsmi,\pricei)  \le 1, & \forall i,\bidsmi  \\[.2in]
   z_i(\bidsmi,\pricei) \ge 0, & \forall i,\bidsmi, \pricei
\end{cases}
\end{equation*}
Note that the summation over $p_i$ in $\LS_1$ is taken in the domain $\support=\big\{(1+\delta)^t ~|~ t=0,1,2,\ldots, N\big\}$.

We remark that the correspondence between $\{z_i(\bidsmi,\pricei)\}_i$ and truthful auctions is not one-to-one but one-to-many.
For any auction, there is a corresponding probability profile $\{z_i(\bidsmi,\pricei)\}_i$ that satisfies $\LS_1$. On the other hand, different auctions may have the same probability profile $\{z_i(\bidsmi,\pricei)\}_i$ (thus, they have the same expected revenue). Note that for any given $\{z_i(\bidsmi,\pricei)\}_i$, we can construct at least one corresponding truthful auction, which independently offers the threshold price $\pricei$ to each bidder $i$ with probability $z_i(\bidsmi,\pricei)$.

We define $$x_i(\bids)=x_i(\bidsmi,b_i)=\sum_{\pricei=1}^{\bidi}\pricei\cdot z_i(\bidsmi,\pricei).$$
Intuitively, $x_i(\bids)$ gives the expected revenue obtained from bidder $i$ when the bid vector is $\bids$.
We further define
\begin{equation*}
w\big((1+\delta)^t\big)=
\begin{cases}
\frac{\delta}{(1+\delta)^{t+1}}, & \quad t=0,1,2,\ldots,N-1 \\
\frac{1}{(1+\delta)^N},& \quad  t=N
\end{cases}
\end{equation*}
We note that $w(\cdot)$ can be viewed as a equal revenue distribution over $\support$, which satisfies the following nice property:
\[\sum_{t=k}^N w\big((1+\delta)^t\big)=\frac{1}{(1+\delta)^k}.\]
Let $\weight(\bids)=\prod_{k=1}^{\nbidder}w(\bid_k)$ and $\weight(\bidsmi) =\prod_{k\neq i}w(\bid_k)$.

Given these definitions, the aforementioned linear system $\LS_1$ can be rewritten as follows.
\begin{eqnarray*}\label{LP:main}
\LS_2: \
\begin{cases}
 \lambda \cdot \sum\limits_{i}  x_i(\bidsmi,\bidi) \geq  f(\bids), & \forall\bids \\[.2in]
 \sum\limits_{t=0}^{N}w\big((1+\delta)^t\big)\cdot x_i\big(\bidsmi,(1+\delta)^t\big) \le 1, & \forall i,\bidsmi \\[.2in]
 x_i\big(\bidsmi,(1+\delta)^t\big) \le x_i\big(\bidsmi,(1+\delta)^{t+1}\big), & \forall i,\bidsmi,t=0,1,2,\ldots,N-1 \\[.2in]
 x_i(\bidsmi,\bidi) \ge 0,  & \forall\bids
\end{cases}
\end{eqnarray*}
The third constraint in $\LS_2$ requires monotonicity of $x_i(\cdot)$, which is a necessary condition for the equivalence of the two linear systems.
Indeed, $z_i(\bidsmi,\bidi)=\frac{1}{\bidi} \big(x_i(\bidsmi,\bidi)-x_i(\bidsmi,\frac{\bidi}{1+\delta})\big)$ must be non-negative, where we denote $x_i(\bidsmi,\frac{1}{1+\delta})=0$. To summarize, we have the following claim.

\begin{proposition}
A monotone function $f$ is $\lambda$-attainable if and only if the linear system $\LS_2$ has a feasible solution $\bx=\{x_i(\bids)\}_{i,\bids}$ (or equivalently, $\LS_1$ has a feasible solution $\bz=\{z_i(\bids)\}_{i,\bids}$).
\end{proposition}

For any set $S\subseteq \support^n$, let $S\proji$ denote the projections of $S$ along the $i$-th coordinate.
Formally,
$$S\proji = \big\{\bidsmi ~|~ \exists b_i \mbox{ s.t. } (\bidsmi, b_i)\in S\big\}.$$
For a non-negative and monotone function $f$, we have the following characterization.

\begin{theorem}\label{th:characterization}
Over a domain $\support$, a non-negative and monotone function $f$ is $\lambda$-attainable if and only if for any upward closed set $S\subset\support^n$,
\begin{equation}\label{eq:rev_condition_old}
\sum_{\bids\in S} \weight(\bids)\cdot f(\bids) \le \lambda\cdot \sum_{i}\sum_{\bidsmi\in S\proji}\weight(\bidsmi).
\end{equation}
\end{theorem}

\begin{proof}
{\em Only if (necessity).} If $f(\cdot)$ is $\lambda$-attainable, then there exists a solution $x_i(\bidsmi,\bidi)$ which  satisfies all constraints in $\LS_2$.
Thus,
\begin{align*}
 \sum_{\bids\in S}\weight(\bids)\cdot f(\bids) & \leq \lambda\cdot \sum_{\bids\in S}\weight(\bids)\sum_i x_i(\bidsmi,\bidi)\\
&=\lambda\cdot \sum_i\sum_{\bidsmi\in S\proji}\weight(\bidsmi)\sum_{\bidi:\bids\in S}w(\bidi)\cdot x_i(\bidsmi,\bidi)\\
&\le \lambda\cdot \sum_i\sum_{\bidsmi\in S\proji}\weight(\bidsmi)\sum_{\bidi \in \support}w(\bidi)\cdot x_i(\bidsmi,\bidi)\\
&\le \lambda\cdot \sum_i\sum_{\bidsmi\in S\proji}\weight(\bidsmi).
\end{align*}

{\em If (sufficiency).} Our goal is to find a feasible solution $\bx=\{x_i(\bids)\}_{i,\bids}$ to the above linear system $\LS_2$, given the system of inequalities \eqref{eq:rev_condition_old}. Our proof is constructive: we provide a procedure of continuously increasing $\{x_i(\bids)\}_{i,\bids}$, starting from 0, to a point where all constraints of $\LS_2$ are satisfied. In order to do that, we write a slightly more general linear system of the following form.
%
%
%
\begin{eqnarray*}
\LS_3: \
\begin{cases}
 \lambda \cdot \sum\limits_{i}  x_i(\bidsmi,\bidi) \geq  f(\bids), & \forall\bids \\[.2in]
 \sum\limits_{t=0}^{N}w\big((1+\delta)^t\big)\cdot x_i\big(\bidsmi,(1+\delta)^t\big) \le g_i(\bidsmi), & \forall i,\bidsmi \\[.2in]
 x_i\big(\bidsmi,(1+\delta)^t\big) \le x_i\big(\bidsmi,(1+\delta)^{t+1}\big), & \forall i,\bidsmi,t=0,1,2,\ldots,N-1 \\[.2in]
 x_i(\bidsmi,\bidi) \ge 0,  & \forall\bids
\end{cases}
\end{eqnarray*}
The only difference between $\LS_2$ and $\LS_3$ is $g_i(\bidsmi)$ in the second constraint.
Intuitively, $g_i(\bidsmi)$ represents a total mass along direction $i$ at the point $\bidsmi$.
Initially, $g_i(\bidsmi)=1$, for all $i$ and $\bidsmi$, and in which case, the two linear systems $\LS_2$ and $\LS_3$ are identical. We will however decrease the values of $g_i(\bidsmi)$'s in the process of the proof while maintaining the following condition, derived from \eqref{eq:rev_condition_old}, for all upward closed sets $S\subset\support^n$.
\begin{equation}\label{eq:rev_condition}
 \sum_{\bids\in S} \weight(\bids)\cdot f(\bids) \le \lambda \cdot\sum_{i}\sum_{\bidsmi\in S\proji} g_i(\bidsmi) \weight(\bidsmi).
\end{equation}

Let $\tight$ be the collection of upward closed sets for which the above inequality \eqref{eq:rev_condition} is tight. It turns out that $\tight$ has a nice structure summarized in the following claim.

\begin{claim}
\label{cl:tight_inequalities}
If $S_1, S_2\in\tight$, then $S_1\cup S_2, S_1\cap S_2\in\tight$.
\end{claim}

\begin{proof}
We note that both $S_1$ and $S_2$ are upward closed sets, then so are $S_1\cup S_2$ and $S_1\cap S_2$. This implies that
\begin{eqnarray}
\sum\limits_{\bids\in S_1\cup S_2} \weight(\bids)\cdot f(\bids) &\le& \lambda\cdot \sum\limits_{i}\sum\limits_{\bidsmi\in (S_1\cup S_2)\proji}g_i(\bidsmi)\cdot\weight(\bidsmi) \label{eq:cup_ineqaulity}\\
 \sum\limits_{\bids\in S_1\cap S_2} \weight(\bids)\cdot f(\bids) &\le& \lambda\cdot \sum\limits_{i}\sum\limits_{\bidsmi\in (S_1\cap S_2)\proji}g_i(\bidsmi)\cdot\weight(\bidsmi) \label{eq:cap_ineqaulity}
\end{eqnarray}

As both $S_1,S_2\in\tight$, we have
\begin{eqnarray*}
 \sum\limits_{\bids\in S_1} \weight(\bids)\cdot f(\bids) &=& \lambda\cdot \sum\limits_{i}\sum\limits_{\bidsmi\in S_1\proji }g_i(\bidsmi)\cdot\weight(\bidsmi)\\
 \sum\limits_{\bids\in S_2} \weight(\bids)\cdot f(\bids) &=& \lambda\cdot  \sum\limits_{i}\sum\limits_{\bidsmi\in S_2\proji }g_i(\bidsmi)\cdot\weight(\bidsmi)
\end{eqnarray*}
For each $\bidsmi$, if $\bidsmi$ is a projection of some $\bids\in S_1\cup S_2$, then either $\bidsmi\in S_1\proji$ or $\bidsmi\in S_2\proji$;
if $\bidsmi$ is a projection of some $\bids\in S_1\cap S_2$, then $\bidsmi\in S_1\proji$ and $\bidsmi\in S_2\proji$.
Hence,
\begin{eqnarray*}
& & \sum\limits_{\bids\in S_1\cup S_2} \weight(\bids)\cdot f(\bids)+ \sum\limits_{\bids\in S_1\cap S_2} \weight(\bids)\cdot f(\bids) \\[.1in]
&=& \sum\limits_{\bids\in S_1} \weight(\bids)\cdot f(\bids) + \sum\limits_{\bids\in S_2} \weight(\bids)\cdot f(\bids)\\
&=& \lambda\cdot \sum\limits_{i} \left(\sum\limits_{\bidsmi\in S_1\proji }g_i(\bidsmi)\cdot\weight(\bidsmi) + \sum\limits_{\bidsmi\in S_2\proji }g_i(\bidsmi)\cdot\weight(\bidsmi) \right) \\
&\ge& \lambda\cdot \sum\limits_{i}\left(\sum\limits_{\bidsmi\in (S_1\cup S_2)\proji}g_i(\bidsmi)\cdot\weight(\bidsmi)+
\lambda\cdot \sum\limits_{\bidsmi\in (S_1\cap S_2)\proji }g_i(\bidsmi)\cdot\weight(\bidsmi)\right)
\end{eqnarray*}

Therefore, both inequalities \eqref{eq:cup_ineqaulity} and \eqref{eq:cap_ineqaulity} are tight. Thus, $S_1\cap S_2, S_1\cup S_2\in\tight$.
\end{proof}

The high-level idea of our proof is to continuously increase $\{x_i(\bids)\}_{i,\bids}$, starting from 0, to a point where all constraints of $\LS_2$ are satisfied. To implement the idea, we identify two special sets $S_0$ and $S_1$, where $S_0\supsetneq S_1$, and a special set $T_i\subseteq S_0\proji \setminus S_1\proji$ for a specific coordinate $i$.
Then for each $\bidsmi \in T_i$, we find a threshold point $c_i(\bidsmi)$ and increase $x_i(\bidsmi,b)$ by $\varepsilon$ for all $b\ge c_i(\bidsmi)$. (In the proof, $c_i(\bidsmi)$ turns out to be the boundary of $S_0$.) 
However, in the process, instead of increasing $x_i$'s, we decrease the values of $f(\cdot)$ and $g_i(\cdot)$ to simplify our analysis (thus, the values of $x_i$'s do not change): For each $\bidsmi \in T_i$, we decrease $f(\bidsmi,b)$ by $\lambda\cdot \varepsilon$ for $b\ge c_i(\bidsmi)$ to have an equivalent effect on the first constraint of $\LS_3$; further, we subtract $\frac{\varepsilon}{c_i(\bidsmi)}$ from $g_i(\bidsmi)$ to balance the update of $x_i$'s in the second constraint of $\LS_3$.
The process continues until all $f(\cdot)$'s become 0, from which point we get an equivalent solution $\{x_i(\bids)\}_{i,\bids}$ to the original problem. We next describe the formal proof.


%
%

\medskip
Let $R = \{\bids ~|~ f(\bids) > 0\}$ be the support of $f(\cdot)$.
Since $f(\cdot)$ is a monotone function, we know that $R$ is an upward closed set. If $R=\emptyset$, then we are done; thus, we assume that $R\neq \emptyset$. For any $S\in \tight$, the following chain of inequalities
\begin{eqnarray*}
\sum_{\bids\in S} \weight(\bids)\cdot f(\bids) = \sum_{\bids\in S\cap R} \weight(\bids)\cdot f(\bids)
\le \lambda\cdot \sum_{i}\sum_{\bidsmi\in (S\cap R)\proji}\weight(\bidsmi)
\le \lambda\cdot \sum_{i}\sum_{\bidsmi\in S\proji}\weight(\bidsmi)
\end{eqnarray*}
is in fact an equality. Hence, the above two inequalities are tight and $S \cap R\in \tight$.
Hence, in the following, we will only consider those tight sets in $\tight$ that are contained in $R$.

Let $S_0=R$. We will maintain a chain of upward closed sets $S_1,\dots,S_m\in\tight$ with the following structure:
\[
R=S_0 \supsetneq S_1\supsetneq S_2\supsetneq\dots\supsetneq S_m=\emptyset.
\]
During our process, we preserve the following key properties ($\dag$):
\begin{center}
\fbox{\parbox{5.7in}{
\hspace{0.05in} \\[-0.05in] Invariant properties ($\dag$)
\begin{enumerate}
  \item Inequality \eqref{eq:rev_condition} holds for all upward closed sets.
  \item All sets $S_1,\ldots,S_m$ in the chain $R=S_0 \supsetneq S_1\supsetneq S_2\supsetneq\dots\supsetneq S_m=\emptyset$ remain tight.
  \item $f(\cdot)$ is nonnegative and monotonically increasing.
  \item $g_i(\bidsmi)$ is nonnegative and monotonically decreasing on $(S_{j-1})\proji \setminus S_j\proji$ for every $1\le i \le n$ and  $1\le j\le m$.
\end{enumerate}
}}
\end{center}

%
%
%

We note that all these properties hold at the beginning of the process (e.g., we may simply choose $m=1$ and $S_m=\emptyset$).
Since $R\supsetneq S_1$, we have
\[\sum_{\bids\in R}\weight(\bids)\cdot f(\bids)>\sum_{\bids\in S_1}\weight(\bids)\cdot f(\bids).\]
Since $R$ satisfies condition (\ref{eq:rev_condition}) and $S_1\in\tight$, we have
\begin{eqnarray*}
 \sum_{\bids\in R} \weight(\bids)\cdot f(\bids) &\leq \lambda\cdot \sum\limits_{i}\sum\limits_{\bidsmi\in R\proji}g_i(\bidsmi)\cdot\weight(\bidsmi)  \\
\sum\limits_{\bids\in S_1} \weight(\bids)\cdot f(\bids) &= \lambda\cdot \sum\limits_{i}\sum\limits_{\bidsmi\in S_1\proji}g_i(\bidsmi)\cdot\weight(\bidsmi).
\end{eqnarray*}
%
%
Hence,
\begin{equation}\label{eq:index}
\sum_{i}\sum_{\bidsmi\in R\proji}g_i(\bidsmi)\cdot\weight(\bidsmi)>\sum\limits_{i}\sum\limits_{\bidsmi\in S_1\proji}g_i(\bidsmi)\cdot\weight(\bidsmi).
\end{equation}
Thus, we can find $i$ and $\bidsmi$ such that $\bidsmi\in R\proji \setminus S_1\proji$ and $g_i(\bidsmi)>0.$
From now on to the end of the proof, for notational simplicity we use $i$ to denote this particular index rather than a generic one.
Let $T_i=\left\{\bidsmi\in R\proji \setminus S_1\proji \mid g_i(\bidsmi)>0\right\}$; note that $T_i\neq\emptyset.$

For each $\bidsmi\in T_i$ we consider the smallest $c_i(\bidsmi)\in\support$ such that $f(\bidsmi,c_i(\bidsmi))\neq 0.$
We note that $c_i(\bidsmi)$ is well defined for each $\bidsmi\in T_i$ since $T_i\subset R\proji$.
For the fixed index $i$, we intend to update $\bx$ as follows:
\[
x_i(\bidsmi,b)\leftarrow x_i(\bidsmi,b)+\varepsilon, \mbox{ for all } \bidsmi\in T_i \mbox{ and } b\ge c_i(\bidsmi),
\]
for some fixed $\varepsilon>0$. In our process, instead of increasing $\bx$ we decrease the values of $f(\cdot)$ and $g_i(\cdot)$ as follows:
\begin{equation*}
(\lozenge) \ \
\begin{cases}
   f(\bidsmi,b) \leftarrow f(\bidsmi,b)-\lambda\cdot \varepsilon, & \mbox{ for all } \bidsmi\in T_i \mbox{ and } b\ge c_i(\bidsmi), \\[.05in]
   g_i(\bidsmi) \leftarrow g_i(\bidsmi)- \frac{\varepsilon}{c_i(\bidsmi)}, & \mbox{ for all } \bidsmi\in T_i.
\end{cases}
\end{equation*}
The decrements are with respect to the fixed index $i$ only\footnote{Because the process updates values only for one dimension at one step, the auction generated by our approach may not be symmetric. See the example in the next section.} and are implemented by continuously increasing $\varepsilon$ from 0 until the value of one of $f(\bids)$ and $g_i(\bidsmi)$ drops down to $0$, or one more inequality \eqref{eq:rev_condition} becomes tight for a new upward closed set.

Before describing how to proceed with the process, we establish some observations for the above updates $(\lozenge)$.

%
%

\begin{claim}\label{cl:tight_sets}
For any $\varepsilon$ and upward closed set $S\subseteq S_1$, the two sides of the inequality \eqref{eq:rev_condition} for $S$ remain unchanged. In particular, the inequality is still tight for all sets $S_1,\ldots,S_m$ in the chain $R=S_0\supsetneq S_1\supsetneq S_2\supsetneq \dots\supsetneq S_m=\emptyset$.
\end{claim}

\begin{proof}
The claim follows trivially since for any $S\subseteq S_1$, none of $f(\bids)$ and $g_i(\bidsmi)$ changes for every $\bids\in S$ and $\bidsmi\in S\proji$.
\end{proof}

\begin{claim}\label{cl:tight_sets1}
For any upward closed set $S\supset R \setminus S_1$ and $\varepsilon$, the condition \eqref{eq:rev_condition} is still satisfied.
\end{claim}

\begin{proof}
For the considered set $S$, the changes of the left-hand side (LHS) and right-hand side (RHS) of \eqref{eq:rev_condition} are as follows.
\begin{eqnarray*}
\text{LHS} &\leftarrow& \text{LHS}-\varepsilon \lambda \cdot\sum_{\bidsmi\in T_i} \weight(\bidsmi)\cdot \sum_{\bidi=c_i(\bidsmi)}^{(1+\delta)^N} w(\bidi) = \text{LHS}-\varepsilon \lambda \cdot\sum_{\bidsmi\in T_i} \weight(\bidsmi)\cdot \frac{1}{c_i(\bidsmi)} \\
\text{RHS} &\leftarrow& \text{RHS}- \lambda \cdot \sum_{\bidsmi\in T_i} \weight(\bidsmi)\cdot \frac{\varepsilon}{c_i(\bidsmi)}
\end{eqnarray*}
Hence, two sides of~\eqref{eq:rev_condition} decrease by exactly the same amount. Therefore, inequality~\eqref{eq:rev_condition} remains true for $S$.
\end{proof}
%
%
%

\begin{claim}\label{cl:fmonotone}
The function $f(\cdot)$ remains monotonically increasing.
\end{claim}
\begin{proof}
Let us assume to the contrary that $f(\cdot)$ becomes non-monotone after one update. We note that all the four key properties $(\dag)$ hold before the update.  Then there must exist a pair of vectors $\wbids\prec\bids$ such that $0\le f(\bids)<f(\wbids)$ after the update.
We have the following observations, where all variables denote their values before the update.
\begin{itemize}
\item Every value of $f(\cdot)$ either remains the same or decreases by $\lambda \cdot \varepsilon$. Thus, $f(\bids)$ decreases and $f(\wbids)$ remains the same.
\item $\bidsmi\in T_i$, since $f(\bids)$ must decrease. Thus, $g_i(\bidsmi)>0$.
\item $\wbids\in R=S_0$, as $f(\wbids)>0$.
\item $\wbidsmi\notin T_i$ (otherwise, we would decrease $f(\wbids)$ by $\lambda\cdot \varepsilon$, since $f(\wbids)>0$).
\item $\wbidsmi\in S_{0}\proji \setminus S_{1}\proji$, since $\wbidsmi\prec\bidsmi\notin S_1\proji$.
\item $g_i(\wbidsmi)=0$, as $\wbidsmi\notin T_i$.
\item $g_i(\bidsmi)\le g_i(\wbidsmi)$, as $\wbidsmi\prec\bidsmi$ and $g_i(\cdot)$ is decreasing on $S_{0}\proji \setminus S_1\proji$.
\end{itemize}
Therefore, we obtain that $g_i(\bidsmi)\le g_i(\wbidsmi)=0$ and $\bidsmi\notin T_i$, a contradiction.
\end{proof}

\begin{claim}\label{cl:gmonotone}
The function $g_i(\cdot)$ remains monotonically decreasing on $(S_{j-1})\proji \setminus S_j\proji$
for each set $S_j$ in the chain $R=S_0 \supsetneq S_1\supsetneq S_2\supsetneq\dots\supsetneq S_m=\emptyset$.
\end{claim}
\begin{proof}
We note that $g_i(\cdot)$ does not change on $(S_{j-1})\proji \setminus S_j\proji$ for every $j>1$.
Thus, we only need to verify the claim for $S_{0}\proji \setminus S_1\proji$.
Assume to the contrary that there exists a pair of vectors $\wbidsmi,\hbidsmi\in S_{0}\proji \setminus S_1\proji$ such that $\wbidsmi\prec\hbidsmi$ and $0\le g_i(\wbidsmi)<g_i(\hbidsmi)$ after the update. We have the following observations.
\begin{itemize}
\item $g_i(\wbidsmi)$ must decrease.
\item $\wbidsmi\in T_i$.
\item There exists $\wbids=(\wbidsmi,\wbid)\in R$, since $\wbidsmi\in T_i$.
\item Let $\hbids=(\hbidsmi,\wbid)$, then $\hbids\in R$, as $\wbids\prec\hbids$ and $R$ is upward closed.
\item $\hbidsmi\in T_i$ since $g_i(\hbidsmi)>0$.
\end{itemize}
We note that for each $\bidsmi\in T_i$, the value of $g_i(\bidsmi)$ decreases by $\frac{\varepsilon}{c_i(\bidsmi)}$. Since $f(\cdot)$ is monotonically
increasing and $c_i(\bidsmi)$ is the smallest number such that $f(\bidsmi,c_i(\bidsmi))>0$, we have $c_i(\hbidsmi)\le c_i(\wbidsmi)$.
Therefore, the decrement of $g_i(\hbidsmi)$, which is $\frac{\varepsilon}{c_i(\hbidsmi)}$, is not smaller than the decrement of
$g_i(\wbidsmi)$, which is $\frac{\varepsilon}{c_i(\wbidsmi)}$. Since before the decrement we have $g_i(\wbidsmi)\ge g_i(\hbidsmi)$, we derive a contradiction.
\end{proof}

Combining Claim~\ref{cl:tight_sets},~\ref{cl:tight_sets1},~\ref{cl:fmonotone} and~\ref{cl:gmonotone} together, we conclude that after update $(\lozenge)$ the aforementioned four key properties $(\dag)$ still hold.
Thus, we continuously increase $\varepsilon$ until a threshold point when one of the three boundary conditions (i.e., $f(\cdot)=0$, $g_i(\cdot)=0$,
or inequality \eqref{eq:rev_condition} becomes tight for a new set) becomes tight, which forbids us from further increasing
$\varepsilon$. If more than one conditions become tight simultaneously, we consider them according to the following order.
\begin{itemize}
\item At a new point $f(\cdot)$ becomes $0$. We then redefine set $R=\{\bids ~|~ f(\bids) >0\}$. Note that by Claim~\ref{cl:fmonotone}, $R$ remains to be an upward closed set. Further, we let $S_i \leftarrow S_i \cap R$ for all sets in the chain $R=S_0 \supsetneq S_1\supsetneq S_2\supsetneq\dots\supsetneq S_m=\emptyset$. (If some of the sets become the same, we contract the chain into a shorter one. With respect to this refined chain, $g_i(\cdot)$ is still monotonically decreasing on each $(S_{j-1})\proji \setminus S_j\proji$.) Then we start over the whole process unless $R=\emptyset$, in which case the process terminates.

\item At a new point $g_i(\cdot)$ becomes $0$. We start over by finding a new index $i$ and set $T_i$ for inequality \eqref{eq:index}. (The set $R$ and the chain remain the same.)

\item There is a new upward closed set, say $S'\subset\support^n$, for which inequality \eqref{eq:rev_condition} becomes tight. Similarly, we assume without loss of generality that $S'\subset R$.
    First, $S'$ cannot contain $R\setminus S_{1}$ due to Claim \ref{cl:tight_sets1}. Second, $S_1$ cannot contain $S'$ due to Claim \ref{cl:tight_sets}.
    By Claim \ref{cl:tight_inequalities}, we know that $S^*\triangleq S'\cup S_{1}$ is also an upward closed tight set. We observe that $S^*=S'\cup S_{1}\supsetneq S_{1}$ and $R=S_{0}\supsetneq S^*$.
    Thus, we can plug $S^*$ into the sequence
    $$R=S_0\supsetneq S^* \supsetneq S_1\supsetneq S_2\supsetneq\dots\supsetneq S_m=\emptyset.$$
    Since $g_i(\cdot)$ is a monotonically decreasing function on set $S_{1}\proji \setminus S_0\proji$, it is a monotonically decreasing function on both $S_{1}\proji \setminus S^*\proji$ and $S^*\proji \setminus S_0\proji$. All the inequalities \eqref{eq:rev_condition} and conditions in $(\lozenge)$ continue to hold with respect to this new chain; hence, we can start over the whole process.
%
%
\end{itemize}

In each step of the update, we either increase the length of the tight set chain or decrease the support of $f(\cdot)$ or $g_i(\cdot)$. Thus,
the process will eventually terminate as there are only finitely many $\bids\in\support^n$ and finitely many upward closed sets.
We note that if there is $\bids\in\support^n$ with $f(\bids)\neq 0$, then not all $g_i(\cdot)$'s can become $0$, because otherwise \eqref{eq:rev_condition}
would be violated for $S=\support^n$. Therefore, the process must terminate when $f(\bids)=0$ for every $\bids\in\support^n$.

Therefore, the respective sequence of updates of $\{x_i(\bids)\}_{i,\bids}$ is a feasible solution to $\LS_2$.
Indeed, we have checked that the first and second constraints in $\LS_2$ are satisfied.
The third and forth constraints hold true, because in every update $(\lozenge)$ we may only increase $x_i(\bids)$, and for each fixed $\bidsmi$ we increase $x_i(\bidsmi,\bid)$ monotonically in $\bid$. Hence, the benchmark function $f(\cdot)$ is attainable.
\end{proof}

We further have the following claim, which will be used in the next section to prove that $\F^{(2)}$ is attainable.

\begin{claim}\label{cl:symmetry}
If $f(\cdot)$ is a symmetric monotone function (i.e., $f(\bids)$ is invariant under permutations of coordinates in $\bids$), then it is sufficient to verify inequality \eqref{eq:rev_condition_old} only for all symmetric upward closed sets for $f(\cdot)$ to be attainable.
\end{claim}

\begin{proof}
Assume that there is an upward closed set that violates \eqref{eq:rev_condition_old}. Then we can continuously decrease $f(\cdot)$ in an arbitrary way so that it remains monotone and symmetric; the process continues until the moment when
\eqref{eq:rev_condition_old} holds for all upward closed sets. Let $f'$ denote the resulting final function; note that $f'$ is attainable. Let $S\neq\emptyset$ be an upward closed set that (i) is tight for $f'$ and (ii) violates inequality \eqref{eq:rev_condition_old} for $f$. We note that such a set $S$ must exist. We then consider upward closed sets $S_1,\dots,S_{n!}$ obtained from $S$ by permuting coordinates of every $\bids\in S$ (in a consistent way). Since $f$ and $f'$ are symmetric functions, $S_1,\dots,S_{n!}$ share properties (i) and (ii) with $S$.

Applying Claim \ref{cl:tight_inequalities} to function $f'$ and tight sets $S_1,\dots,S_{n!}$, we obtain that $S_1\cup\dots\cup S_{n!}$ is tight for $f'$. Furthermore, by the definition of $S$, there is $\bids\in S$ such that $f(\bids)>f'(\bids)$. We conclude that $S_1\cup\dots\cup S_{n!}$ violates \eqref{eq:rev_condition_old} for $f$. Hence, if $f(\cdot)$ is not attainable, then we can find a symmetric upward closed set that violates condition \eqref{eq:rev_condition_old}.
\end{proof}

\subsection{Unbounded and Continuous Domains}

Our above analysis works for the discrete and bounded domain where all bids are in the set
$\support = \big\{(1+\delta)^t ~|~ t=0,1,2,\ldots, N\big\}$. Our analysis continues to hold when $N=\infty$. Indeed, the argument for the necessity part of Theorem~\ref{th:characterization} for $N=\infty$ works in exactly the same way as the finite and bounded case. For the sufficiency part, we apply a standard argument from measure theory and mathematical analysis as follows.

For the unbounded domain, if inequality \eqref{eq:rev_condition_old} in the characterization theorem holds true for all upward closed sets,
then every finite version of $\LS_2$ must have a feasible solution $\bx_{_N}$ for every finite size $N$ of the support $\support(N)=\{1,1+\delta,\ldots,(1+\delta)^N\}$, where all outside densities are accumulated at the boundary of $\support(N)$. We consider a corresponding feasible solution $\bz_{_N}=\{z_i(\bids)\}_{i,\bids}$ to $\LS_1$ for each $\support(N)$. We treat the price distribution $z_i(\bidsmi)$ for each bidder $i$ as a measure defined on the $\sigma$-algebra generated by the sets $\{[1,1+\delta),\ldots, [(1+\delta)^N,\infty)\}^n$. Every solution $\bz_{_N}$ to $\LS_1$ is a vector of $n$ measures $z_i(\bids)$, of which we altogether further regard as a measure defined on the $\sigma$-algebra $\sigma_N$ generated by the sets  $\{[1,1+\delta),\ldots, [(1+\delta)^N,\infty)\}^{n\times n}$. In a few steps below we show how to construct a measure $\bz^*$ that generates a feasible solution of $\LS_1$ on every $\sigma$-algebra $\sigma_N$ for every $N$.

\begin{enumerate}
\item We notice that for any $N_2>N_1$, every solution $\bz_{N_2}$ to $\LS_1$ for $N_2$ is also a solution of $\LS_1$ for $N_1$. Indeed, as a solution of $\LS_1$ for $N_1$ we may just use the same measure $\bz_{N_2}$ on a smaller $\sigma$-algebra $\sigma_{N_1}$.
\item We observe that the set of all feasible solutions $\bz_N$ to $\LS_1$ for any fixed $N$ is a compact set (bounded and closed). We recall that one of the definitions of a compact set says that every infinite sequence $\bz_{_N}$ $(N=1,2,\ldots)$ of points in a compact $C$ must contain an infinite subsequence $\bz_{_{N(j)}}$ $(j=1,2,\ldots )$ that converges to a point $\bbarz_{_N}\in C$.
\item For every fixed $N=\ell$, we consider a sequence of solutions $\bz_k$ to $\LS_1$ for each $k=\ell,\ell+1,\dots$. For each $\bz_k$ we get a feasible solution $\bz_k|_\ell$ to $\LS_1$ for $N=\ell$. Further, since the set of the solutions to $\LS_1$ for $N=\ell$ is a compact set, we may choose an infinite subsequence that converges to $\bbarz_\ell$.
\item  We note that by our construction, the measure $\bbarz_\ell$ can be extended to a feasible solution of $\LS_1$ for every $k\ge\ell$, i.e., we can find a feasible solution $\bhatz_k$ to $\LS_1$ for $N=k$ such that as a measure $\bbarz_\ell=\bhatz_k|_\ell$. Indeed, we can take the infinite sequence $\{N(j)\}$ for which $\bz_{_{N(j)}}|_\ell\to\bbarz_\ell$ and consider another infinite sequence of measures $\bz_{_{N(j)}}|_k$. Within the latter sequence we can choose an infinite subsequence converging to a feasible solution $\bhatz_k$ of $\LS_1$ for $N=k$ (note that $\bhatz_k|_\ell=\bbarz_\ell$).
\item In fact, when constructing each $\bbarz_\ell$ we can ensure that $\bbarz_\ell|_j =\bbarz_j$ for every $j\le\ell.$
\item Finally, we may define our measure $\bz^*$ on the unbounded domain $\Lambda$ as a limit of $\bbarz_\ell$, where $\ell\to\infty$.
\end{enumerate}

In a similar way by taking $\delta\rightarrow 0$, we can extend Theorem~\ref{th:characterization} to the case of continuous support. Hence, in the continuous and unbounded domain, the sufficient and necessary condition \eqref{eq:rev_condition_old} in Theorem~\ref{th:characterization} translates into the following.

\begin{theorem}\label{th:characterization-unbounded}
A monotone function $f(\cdot)$ is $\lambda$-attainable if and only if for any measurable upward closed set $S\subset\R^{n}_{\ge 1}$
\begin{equation}\label{eq:condition}
\int\limits_{S}f(\bids)\cdot\weight(\bids)\diff\bids\le \lambda\cdot\sum_{i=1}^n \ \int\limits_{\ S\proji}\weight(\bidsmi)\diff\bidsmi,
\end{equation}
where
\begin{equation*}
 w(t)= \frac{1}{t^2},\quad\quad\weight(\bids)=\prod_{k=1}^{\nbidder}w(\bid_k),\quad\quad\weight(\bidsmi) = \prod\limits_{k\neq i}w(\bid_k).
\end{equation*}
\end{theorem}

\section{A Simple Example}


We illustrate in this section how our process in the proof of Theorem~\ref{th:characterization} works on a simple example.
We take $n=2$, support $\support=\{1,2\}$, and consider the following benchmark:
\[ f(b_1,b_2)=\begin{cases} 1.5 & \min(b_1, b_2)<2 \\ 3.5  & \mbox{otherwise} \end{cases} \]
Note that for this specific benchmark, the optimal competitive ratio $\lambda$ is $1$ for which the condition \eqref{eq:rev_condition_old} holds for all upward closed sets.

To simplify the presentation, we use the following tables to denote the values of $x_1(\bid_1,\bid_2)$, $x_2(\bid_1,\bid_2)$, $f(\bid_1,\bid_2)$, $g_1(\bidsmi[1])$ and $g_2(\bidsmi[2])$, respectively.
\begin{center}
\begin{tabular}{|c|c|}
  \hline
    $x_1(1,2)$ & $x_1(2,2)$ \\
    \hline
    $x_1(1,1)$ & $x_1(2,1)$ \\
  \hline
\end{tabular}\quad
\begin{tabular}{|c|c|}
  \hline
    $x_2(1,2)$ & $x_2(2,2)$ \\
    \hline
    $x_2(1,1)$ & $x_2(2,1)$ \\
  \hline
\end{tabular}\quad
\begin{tabular}{|c|c|}
  \hline
    $f(1,2)$ & $f(2,2)$ \\
    \hline
    $f(1,1)$ & $f(2,1)$ \\
  \hline
\end{tabular}\quad
\begin{tabular}{|l|c|}
    \cline{1-1} \mbox{ $g_1(2)$} \\ \hline \slashbox {\mbox{ $g_1(1)$}}{$g_2(1)$} & \raisebox{-1.4ex}[0pt]{$g_2(2)$} \\ \hline
\end{tabular}
\end{center}
We write $\LS_3$ for our specific support $\support$ as follows.
\begin{eqnarray*}
\LS_3: \
\begin{cases}
 x_1(\bids)+x_2(\bids) \geq  f(\bids), & \forall\bids \\[.2in]
 w(1)\cdot x_i(\bidsmi,1)+w(2)\cdot x_i(\bidsmi,2)\le g_i(\bidsmi), & \forall i,\bidsmi \\[.2in]
 x_i(\bidsmi,1) \le x_i(\bidsmi,2), & \forall i,\bidsmi \\[.2in]
 x_i(\bids) \ge 0.  & \forall i,\bids,
\end{cases}
\end{eqnarray*}
where $w(2)=\frac{1}{2}$ and $w(1)=1-w(2)=\frac{1}{2}$.

The initial values of the tables are as follows (we record $x_i$'s so that we can reconstruct the mechanism).
\begin{center}
\begin{tabular}{|c@{\hspace{6pt}}c@{\hspace{0pt}}}
 \hline
  $x_1$
  &
  \begin{tabular}{||c|c|}
    $0$ & $0$\\
    \hline
    $0$ & $0$
  \end{tabular}\\
 \hline
\end{tabular}\quad\quad
\begin{tabular}{|c@{\hspace{6pt}}c@{\hspace{0pt}}}
 \hline
  $x_2$
  &
  \begin{tabular}{||c|c|}
    $0$ & $0$\\
    \hline
    $0$ & $0$
  \end{tabular}\\
 \hline
\end{tabular}\quad\quad
\begin{tabular}{|c@{\hspace{6pt}}c@{\hspace{0pt}}}
 \hline
  $f$
  &
  \begin{tabular}{||c|c|}
    $1.5$ & $3.5$\\
    \hline
    $1.5$ & $1.5$
  \end{tabular}\\
 \hline
\end{tabular}\quad\quad
\begin{tabular}{|l|c|}
    \cline{1-1} \mbox{ $1$} \\ \hline \slashbox {\mbox{ $1$}}{$1$} & \raisebox{-1.4ex}[0pt]{$1$} \\ \hline
\end{tabular}
\end{center}

We begin our process with the trivial chain $\{1,2\}^{\otimes 2}=R=S_0\supsetneq S_1=\emptyset$. We choose $i=1$, then according to our construction, $T_i=\{1,2\}$ and $c_i(\bidsmi)=1$ for each $\bidsmi\in T_i$. We continuously increase $\varepsilon$ up to $0.5$ until two new sets $\{(2,1),(2,2)\}$ and $\{(2,2)\}$ become tight simultaneously. After the update, we get the following tables.
\begin{center}
\begin{tabular}{|c@{\hspace{6pt}}c@{\hspace{0pt}}}
 \hline
  $x_1$
  &
  \begin{tabular}{||c|c|}
    $0.5$ & $0.5$\\
    \hline
    $0.5$ & $0.5$
  \end{tabular}\\
 \hline
\end{tabular}\quad\quad
\begin{tabular}{|c@{\hspace{6pt}}c@{\hspace{0pt}}}
 \hline
  $x_2$
  &
  \begin{tabular}{||c|c|}
    $0$ & $0$\\
    \hline
    $0$ & $0$
  \end{tabular}\\
 \hline
\end{tabular}\quad\quad
\begin{tabular}{|c@{\hspace{6pt}}c@{\hspace{0pt}}}
 \hline
  $f$
  &
  \begin{tabular}{||c|c|}
    $1$ & $3$\\
    \hline
    $1$ & $1$
  \end{tabular}\\
 \hline
\end{tabular}\quad\quad
\begin{tabular}{|l|c|}
    \cline{1-1} \mbox{ $0.5$} \\ \hline \slashbox {\mbox{ $0.5$}}{$1$} & \raisebox{-1.4ex}[0pt]{$1$} \\ \hline
\end{tabular}
\end{center}

We add these two new sets to our chain so that now it looks as $\{1,2\}^{\otimes 2}=R=S_0\supsetneq S_1\supsetneq S_2\supsetneq S_3=\emptyset$, where $S_1=\{(2,1),(2,2)\}$ and $S_2=\{(2,2)\}$. Now since $R\proji[1] \setminus S_1\proji[1]$ is an empty set, we must choose $i=2$; thus, $T_i=\{1\}$ and $c_i(\bidsmi)=1.$
We continuously increase $\varepsilon$ up to $1$ until all $f(1,1),f(1,2),g_2(1)$ simultaneously become zero. We get the following tables after the update.


\begin{center}
\begin{tabular}{|c@{\hspace{6pt}}c@{\hspace{0pt}}}
 \hline
  $x_1$
  &
  \begin{tabular}{||c|c|}
    $0.5$ & $0.5$\\
    \hline
    $0.5$ & $0.5$
  \end{tabular}\\
 \hline
\end{tabular}\quad\quad
\begin{tabular}{|c@{\hspace{6pt}}c@{\hspace{0pt}}}
 \hline
  $x_2$
  &
  \begin{tabular}{||c|c|}
    $1$ & $0$\\
    \hline
    $1$ & $0$
  \end{tabular}\\
 \hline
\end{tabular}\quad\quad
\begin{tabular}{|c@{\hspace{6pt}}c@{\hspace{0pt}}}
 \hline
  $f$
  &
  \begin{tabular}{||c|c|}
    $0$ & $3$\\
    \hline
    $0$ & $1$
  \end{tabular}\\
 \hline
\end{tabular}\quad\quad
\begin{tabular}{|l|c|}
    \cline{1-1} \mbox{ $0.5$} \\ \hline \slashbox {\mbox{ $0.5$}}{$0$} & \raisebox{-1.4ex}[0pt]{$1$} \\ \hline
\end{tabular}
\end{center}

We update the chain so that $\{(2,1), (2,2)\}=R=S_0\supsetneq S_1\supsetneq S_2=\emptyset$, where $S_1=\{(2,2)\}$. (The chain becomes shorter, because $R$ has decreased.) We choose $i=1$, and get $T_i=\{1\}$ and $c_i(\bidsmi)=2$ for $\bidsmi[1] = 1$. We continuously increase $\varepsilon$ to $1$ until both $f(2,1)$ and $g_1(1)$ become zero simultaneously. After the update, we have
%
%
\begin{center}
\begin{tabular}{|c@{\hspace{6pt}}c@{\hspace{0pt}}}
 \hline
  $x_1$
  &
  \begin{tabular}{||c|c|}
    $0.5$ & $0.5$\\
    \hline
    $0.5$ & $1.5$
  \end{tabular}\\
 \hline
\end{tabular}\quad\quad
\begin{tabular}{|c@{\hspace{6pt}}c@{\hspace{0pt}}}
 \hline
  $x_2$
  &
  \begin{tabular}{||c|c|}
    $1$ & $0$\\
    \hline
    $1$ & $0$
  \end{tabular}\\
 \hline
\end{tabular}\quad\quad
\begin{tabular}{|c@{\hspace{6pt}}c@{\hspace{0pt}}}
 \hline
  $f$
  &
  \begin{tabular}{||c|c|}
    $0$ & $3$\\
    \hline
    $0$ & $0$
  \end{tabular}\\
 \hline
\end{tabular}\quad\quad
\begin{tabular}{|l|c|}
    \cline{1-1} \mbox{ $0.5$} \\ \hline \slashbox {\mbox{ $0$}}{$0$} & \raisebox{-1.4ex}[0pt]{$1$} \\ \hline
\end{tabular}
\end{center}

We update the chain so that $\{(2,2)\}=R=S_0\supsetneq S_1=\emptyset$. We choose $i=1$, and get $\bidsmi=2$ and $c_i(\bidsmi)=2$. We continuously increase $\varepsilon$ up to $1$ until $g_1(2)$ becomes zero. Now we have
\begin{center}
\begin{tabular}{|c@{\hspace{6pt}}c@{\hspace{0pt}}}
 \hline
  $x_1$
  &
  \begin{tabular}{||c|c|}
    $0.5$ & $1.5$\\
    \hline
    $0.5$ & $1.5$
  \end{tabular}\\
 \hline
\end{tabular}\quad\quad
\begin{tabular}{|c@{\hspace{6pt}}c@{\hspace{0pt}}}
 \hline
  $x_2$
  &
  \begin{tabular}{||c|c|}
    $1$ & $0$\\
    \hline
    $1$ & $0$
  \end{tabular}\\
 \hline
\end{tabular}\quad\quad
\begin{tabular}{|c@{\hspace{6pt}}c@{\hspace{0pt}}}
 \hline
  $f$
  &
  \begin{tabular}{||c|c|}
    $0$ & $2$\\
    \hline
    $0$ & $0$
  \end{tabular}\\
 \hline
\end{tabular}\quad\quad
\begin{tabular}{|l|c|}
    \cline{1-1} \mbox{ $0$} \\ \hline \slashbox {\mbox{ $0$}}{$0$} & \raisebox{-1.4ex}[0pt]{$1$} \\ \hline
\end{tabular}
\end{center}

The chain $\{(2,2)\}=R=S_0\supsetneq S_1=\emptyset$ remains the same. Finally, we take $i=2$, and get $\bidsmi=2$ and $c_i(\bidsmi)=2$. We continuously increase $\varepsilon$ up to $2$ until $f(2,2)$ and $g_2(2)$ become zero simultaneously. The final tables are the following.
\begin{center}
\begin{tabular}{|c@{\hspace{6pt}}c@{\hspace{0pt}}}
 \hline
  $x_1$
  &
  \begin{tabular}{||c|c|}
    $0.5$ & $1.5$\\
    \hline
    $0.5$ & $1.5$
  \end{tabular}\\
 \hline
\end{tabular}\quad\quad
\begin{tabular}{|c@{\hspace{6pt}}c@{\hspace{0pt}}}
 \hline
  $x_2$
  &
  \begin{tabular}{||c|c|}
    $1$ & $2$\\
    \hline
    $1$ & $0$
  \end{tabular}\\
 \hline
\end{tabular}\quad\quad
\begin{tabular}{|c@{\hspace{6pt}}c@{\hspace{0pt}}}
 \hline
  $f$
  &
  \begin{tabular}{||c|c|}
    $0$ & $0$\\
    \hline
    $0$ & $0$
  \end{tabular}\\
 \hline
\end{tabular}\quad\quad
\begin{tabular}{|l|c|}
    \cline{1-1} \mbox{ $0$} \\ \hline \slashbox {\mbox{ $0$}}{$0$} & \raisebox{-1.4ex}[0pt]{$0$} \\ \hline
\end{tabular}
\end{center}

Hence, we found a feasible solution $\bx=\{x_i(\bids)\}_{i,\bids}$ to $\LS_2$. The respective solution
$\bz=\{z_i(\bids)\}_{i,\bids}$ to $\LS_1$ looks as follows.
\begin{center}
\begin{tabular}{|c@{\hspace{6pt}}c@{\hspace{0pt}}}
 \hline
  $z_1$
  &
  \begin{tabular}{||c|c|}
    $0.5$ & $0.5$\\
    \hline
    $0.5$ & $0.5$
  \end{tabular}\\
 \hline
\end{tabular}\quad\quad
\begin{tabular}{|c@{\hspace{6pt}}c@{\hspace{0pt}}}
 \hline
  $z_2$
  &
  \begin{tabular}{||c|c|}
    $0$ & $1$\\
    \hline
    $1$ & $0$
  \end{tabular}\\
 \hline
\end{tabular}
\end{center}
In terms of the language of an auction, this reads that the auctioneer offers a random price of $1$ or $2$ to the first bidder and a price that is equal to $\bid_1$ to the second bidder. We note that the resulting auction is not unique due to multiple choices of the coordinate $i$ at every step of our process. It is interesting to notice that we get an asymmetric mechanism although the benchmark is symmetric.

%
%



\section{$\F^{(2)}$ is Attainable}

We show in this section that $\F^{(2)}$ is $\lambda_n$-attainable, for $\lambda_n$ being the number conjectured in~\cite{GoldbergHKSW06}.

\begin{theorem}\label{thm:F2achievable}
$\F^{(2)}(\cdot)$ is $\lambda_n$-attainable for any $n\ge 2$, where
\[\lambda_n = 1-\sum_{i=2}^n \left(\frac{-1}{n}\right)^{i-1} \frac{i}{i-1} {n-1 \choose i-1}.\]
\end{theorem}

For the sake of notational convenience, we employ $F_n(\bids)$ to denote $\F^{(2)}(\bids)$ when the input vector $\bids$ contains $n$ bids.
We let $\support=[1, +\infty)$ and $\support^n=\R^{n}_{\ge 1}$ denote the support of bid vectors.
Our proof is by induction on $n$. The base case $n=2$ is trivial (where $\lambda_2=2$ and $F_2(\cdot)$ is 2-attainable given by the Vickrey auction).
We assume that for all $k\le n$, each $F_k(\cdot)$ is $\lambda_k$-attainable.
In the rest of this section, we show that $F_{n+1}(\cdot)$ is $\lambda_{n+1}$-attainable.

For the equal revenue distribution with density $\weight(\cdot)$, it was shown in \cite{GoldbergHKS04} that
\begin{equation}\label{eq:F2expectation}
\mathbb{E}_{\bids\sim \weight}\big[F_k(\bids)\big] = \int\limits_{\support^{k}}F_k(\bids)\cdot\weight(\bids)\diff\bids=\lambda_k \cdot k.
\end{equation}
We note that $1=\int_{1}^{\infty}w(t)\diff t$ and, therefore, we have
\begin{equation}\label{eq:F2tightset}
\int\limits_{\support^k}F_k(\bids)\cdot\weight(\bids)\diff\bids=\lambda_k\cdot \sum_{i=1}^{k}\ \int\limits_{\ \support^{k-1}}\weight(\bidsmi)\diff\bidsmi.
\end{equation}
Comparing to the inequality \eqref{eq:condition} in Theorem~\ref{th:characterization-unbounded}, we have the following claim.

\begin{claim}
\label{cl:support_tight}
For $k=2,\ldots,n$, each $\lambda_k$ is chosen so that inequality \eqref{eq:condition} is tight for the set $\support^k$.
\end{claim}

Let $$G_n(\bz)=F_{n+1}(1,\bz), \quad\forall \bz\in\support^n.$$
By the following claim, to prove that $F_{n+1}(\bids)$ with $\bids\in\support^{n+1}$ is $\lambda_{n+1}$-attainable,
it suffices to show that $G_n(\bz)$ with $\bz\in\support^n$ is $\lambda_{n+1}$-attainable.

\begin{claim}\label{cl:G_n}
If $G_n(\bz)$ with $\bz\in\support^n$ is $\lambda_{n+1}$-attainable, then $F_{n+1}(\bids)$ with $\bids\in\support^{n+1}$ is $\lambda_{n+1}$-attainable.
\end{claim}

\begin{proof}
To prove the claim, we convert a $\lambda_{n+1}$-competitive auction $\auction$ for $G_n(\bz)$ to one for $F_{n+1}(\bids)$.
For each bidder $i$ and $\bids\in\support^{n+1}$, we may find a coordinate $i^*\neq i$ with the smallest bid $b_{i^*}=\min_{j\neq i}\{\bid_j\}$. If the minimum is not unique, we take $i^*$ with largest index. Our auction for $F_{n+1}(\bids)$ offers price $p_i\cdot b_{i^*}$ to bidder $i$ when observing $\bidsmi$, where $p_i$ is the price offered to $i$ in auction $\auction$ when observing bids $\frac{\bids_{(-i,-i^*)}}{b_{i^*}}$ ($\bids_{(-i,-i^*)}$ is the bid vector of length $n-1$ obtained from $\bids$ by removing $i$ and $i^*$). We shall prove that this auction has expected revenue of at least $\frac{F_{n+1}(\bids)}{\lambda_{n+1}}$.

First, we notice that $F_k(\bids)$ scales linearly with $\bids$, i.e.,
$F_k(t\cdot \bids)=t\cdot F_k(\bids)$ for any $t\in\R$. Let $j^*$ be the coordinate with the smallest bid $b_{j^*}=\min_{1\le j\le n+1}\{\bid_j\}$. We take $j^*$ with the largest index if the minimum is not unique. Then for any bidder $i\neq j^*$, $i^*=j^*$ in the above construction. We will prove that the expected revenue from all bidders except $j^*$ is at least $\frac{F_{n+1}(\bids)}{\lambda_{n+1}}$. The expected revenue from all bidders except $j^*$ in our auction for $F_{n+1}(\bids)$ is exactly $b_{j^*}$ times the expected revenue of $\auction$ on the bid vector $\frac{\bids_{-j^*}}{b_{j^*}}$. The claim follows from the following relation between $G_n(\bz)$ and $F_{n+1}(\bids)$
\[F_{n+1}(\bids)= b_{j^*} \cdot F_{n+1}\left(1, \frac{\bids_{-j^*}}{b_{j^*}}\right)= b_{j^*} \cdot G_n\left(\frac{\bids_{-j^*}}{b_{j^*}}\right).\]
%
%
%
\end{proof}

We now show that $G_n(\bz)$ behaves similarly as $F_{n+1}(\bids)$ in \eqref{eq:F2expectation}.

\begin{claim}\label{cl:Gntight}
The set $\support^n$ is tight for the benchmark $G_n(\bz)$ and competitive ratio $\lambda_{n+1}$ in~\eqref{eq:rev_condition_old}, i.e.,
\[
\int\limits_{\support^n}G_n(\bz)\cdot\weight(\bz)\diff\bz = \lambda_{n+1}\cdot n.
\]
\end{claim}

\begin{proof}
Given the equation \eqref{eq:F2expectation}, we know that
\begin{equation*}
\lambda_{n+1}\cdot(n+1)=\int\limits_{\support^{n+1}}F_{n+1}(\bids)\weight(\bids)\diff\bids=
\sum_{i=1}^{n+1}\ \int\limits_{1}^{\infty}w(\bidi)\int\limits_{\bidi\cdot\support^n}F_{n+1}(\bids)\weight(\bidsmi)\diff\bidsmi\diff\bidi.
\end{equation*}
The last equality holds true because $\support^{n+1}$ can be divided into $n+1$ disjoint sets, where the $i$-th set contains those $\bids$'s
with the smallest coordinate $\bidi$ and $\bidsmi\in\bidi\cdot\support^n$. We further write
\begin{equation*}
\lambda_{n+1}\cdot(n+1)=(n+1)\cdot\int\limits_{1}^{\infty}w(t)\int\limits_{t\cdot\support^n}F_{n+1}(t,\bidsmi[(n+1)])\weight(\bidsmi[(n+1)])\diff\bidsmi[(n+1)]\diff t.
\end{equation*}
Recall that $w(t)= \frac{1}{t^2}$ and $\weight(\bids)=\prod_{i}w(\bid_i)$.
To simplify notation in the next expression, we let $\bx=\bidsmi[(n+1)]$ and $\bz=\frac{1}{t}\cdot\bx$.
\begin{eqnarray*}
\lambda_{n+1} &=& \int\limits_{1}^{\infty}w(t)\int\limits_{t\cdot\support^n}F_{n+1}(t,\bx)\weight(\bx)\diff\bx\diff t\\
&=& \int\limits_{1}^{\infty}t^{-2}\int\limits_{\support^n}t\cdot F_{n+1}(1,\bz)\cdot t^{-2n}\weight(\bz)\cdot t^{n}\diff\bz\diff t\\
&=& \int\limits_{1}^{\infty}t^{-(n+1)}\diff t\int\limits_{\support^n}F_{n+1}(1,\bz)\cdot\weight(\bz)\diff\bz\\
&=& \frac{1}{n}\cdot\int\limits_{\support^n}G_n(\bz)\cdot\weight(\bz)\diff\bz.
\end{eqnarray*}
Thus, the claim follows.
\end{proof}

Let
$$H_n(\bz)=\max\big(0,n+1-F_n(\bz)\big), \quad\forall \bz\in\support^n.$$
We observe that $G_n(\bz)=F_{n+1}(1,\bz)=\max(n+1,F_n(\bz))$ for $\bz\in\support^n$;
thus, $G_n(\bz)=F_n(\bz)+H_n(\bz)$. We note that due to the induction hypothesis, $F_n(\bz)$ is $\lambda_n$-attainable.
This means that for every symmetric upward closed set $S\subset\support^n$, we have
\begin{equation*}
\int\limits_{S}F_n(\bz)\cdot\weight(\bz)\diff\bz\le \lambda_n\cdot\sum_{i=1}^n \ \int\limits_{S\proji}\weight(\bzmi)\diff\bzmi.
\end{equation*}

We show in Claim~\ref{cl:Hn} below that for every symmetric upward closed set $S\subset\support^n$,
\[
\int\limits_{S}H_n(\bz)\cdot\weight(\bz)\diff\bz\le (\lambda_{n+1}-\lambda_n)\cdot\sum_{i=1}^n\ \int\limits_{S\proji}\weight(\bzmi)\diff\bzmi.
\]
If we combine the above two inequalities, then we get for every symmetric upward closed set $S\subset\support^n$,
\begin{equation*}
\int\limits_{S}G_n(\bz)\cdot\weight(\bz)\diff\bz\le \lambda_{n+1} \cdot\sum_{i=1}^n\ \int\limits_{S\proji}\weight(\bzmi)\diff\bzmi.
\end{equation*}
We observe that $F_n(\cdot), G_n(\cdot)$ and $H_n(\cdot)$ all are symmetric functions.
By Theorem~\ref{cl:symmetry}, we know that $G_n(\bz)$ is $\lambda_{n+1}$-attainable.
Hence, we are only left to prove the following Claim~\ref{cl:Hn} in order to show $\lambda_{n+1}$-attainability of $F_{n+1}(\bids)$
and complete the proof of Theorem \ref{thm:F2achievable}.

%
%

\begin{claim}\label{cl:Hn}
For every symmetric upward closed set $S\subset\support^n$,
\begin{equation}\label{eq:h_condition}
\int\limits_{S}H_n(\bz)\cdot\weight(\bz)\diff\bz\le (\lambda_{n+1}-\lambda_n)\cdot\sum_{i=1}^n\ \int\limits_{S\proji}\weight(\bzmi)\diff\bzmi.
\end{equation}
\end{claim}

\begin{proof}
Since $S$ is a symmetric set, all projections $S\proji$ are in fact equal to the same set, denoted by $X\subset\support^{n-1}.$
We notice that $S$ can be divided into $n$ (almost) disjoint sets $S=S_1\cup\cdots\cup S_n$, where each $S_i=\big\{(\bzmi,z_i) ~|~ \bzmi\in X, z_i\ge\max(\bzmi)\big\}$ and $\max(\bzmi)$ denotes the largest value of all coordinates of $\bzmi$.
Furthermore, note that for any $z_i\in[\max(\bzmi),\infty)$, the value of the benchmark $F_n(\bzmi,z_i)$ is fixed; thus, the value of $H_n(\bzmi,z_i)$ is also fixed.
Hence, we have the following equalities.
\begin{eqnarray}
\label{eq:LHS_h_condition}
\int\limits_{S}H_n(\bz)\cdot\weight(\bz)\diff\bz&=&\sum_{i=1}^{n}\ \int\limits_{S_i}H_n(\bz)\cdot\weight(\bz)\diff\bz\notag\\
&=& n\cdot\int\limits_{X} \weight(\bzmi[n])\int\limits_{\max(\bzmi[n])}^{\infty}H_n(\bzmi[n],t)\cdot w(t)\diff t\diff\bzmi[n]\notag\\
&=& n\cdot\int\limits_{X} \frac{H_n(\bzmi[n],\max(\bzmi[n]))}{\max(\bzmi[n])}\weight(\bzmi[n])\diff\bzmi[n].
\end{eqnarray}

By Claim \ref{cl:Gntight}, we know that
$$\int\limits_{\support^n}G_n(\bz)\cdot\weight(\bz)\diff\bz = \lambda_{n+1}\cdot n = \lambda_{n+1}\cdot\sum_{i=1}^n \ \int\limits_{\support^{n-1}}\weight(\bzmi)\diff\bzmi.$$
Further, equality \eqref{eq:F2tightset} says
$$\int\limits_{\support^n}F_n(\bz)\cdot\weight(\bz)\diff\bz = \lambda_n\cdot\sum_{i=1}^n \ \int\limits_{\support^{n-1}}\weight(\bzmi)\diff\bzmi.$$
Thus,
$$\int\limits_{\support^n}H_n(\bz)\cdot\weight(\bz)\diff\bz = \int\limits_{\support^n}\big(G_n(\bz)-F_n(\bz)\big)\cdot\weight(\bz)\diff\bz = (\lambda_{n+1}-\lambda_n)\cdot\sum_{i=1}^n\ \int\limits_{\support^{n-1}}\weight(\bzmi)\diff\bzmi.$$
That is, $\support^n$ is a tight set for \eqref{eq:h_condition}. This implies that
\begin{eqnarray}
\label{eq:h_tight}
(\lambda_{n+1}-\lambda_n)\cdot n&=&(\lambda_{n+1}-\lambda_n)\cdot \sum_{i=1}^n\ \int\limits_{\support^{n-1}}\weight(\bzmi)\diff\bzmi=
\int\limits_{\support^n}H_n(\bz)\cdot\weight(\bz)\diff\bz  \notag \\
&=&n\cdot\int\limits_{\support^{n-1}} \frac{H_n(\bzmi[n],\max(\bzmi[n]))}{\max(\bzmi[n])}\weight(\bzmi[n])\diff\bzmi[n]
\end{eqnarray}
Therefore,
\begin{eqnarray}
\label{eq:RHS_h_condition}
(\lambda_{n+1}-\lambda_n)\cdot\sum_{i=1}^n\ \int\limits_{S\proji}\weight(\bzmi)\diff\bzmi &=&
(\lambda_{n+1}-\lambda_n)\cdot n\cdot \int\limits_{X}\weight(\bzmi[n])\diff\bzmi[n]\\
&=& n\cdot \int\limits_{\support^{n-1}} \frac{H_n(\bzmi[n],\max(\bzmi[n]))}{\max(\bzmi[n])}\weight(\bzmi[n])\diff\bzmi[n]
\cdot \int\limits_{X}\weight(\bzmi[n])\diff\bzmi[n]\notag
\end{eqnarray}

Combining \eqref{eq:LHS_h_condition} and \eqref{eq:RHS_h_condition}, we are left to show that
\begin{equation}
\label{eq:LHS_h_RHS}
\int\limits_{X} \frac{H_n(\bzmi[n],\max(\bzmi[n]))}{\max(\bzmi[n])}\weight(\bzmi[n])\diff\bzmi[n]
\le
\int\limits_{\support^{n-1}} \frac{H_n(\bzmi[n],\max(\bzmi[n]))}{\max(\bzmi[n])}\weight(\bzmi[n])\diff\bzmi[n]
\cdot \int\limits_{X}\weight(\bzmi[n])\diff\bzmi[n]
\end{equation}
Since $H_n(\bz)$ is a non-negative monotonically decreasing function, the function
$H'(\bzmi[n])\triangleq \frac{H_n(\bzmi[n],\max(\bzmi[n]))}{\max(\bzmi[n])}$
is also non-negative monotonically decreasing. Let $\ind_X(\bzmi[n])$ denote the characteristic function of the set $X$.
Note that $\ind_X(\cdot)$ is a non-negative monotonically increasing function. Now the inequality \eqref{eq:LHS_h_RHS}
simply reads as
\[
\int\limits_{\support^{n-1}} H'(\bzmi[n])\cdot\ind_X(\bzmi[n])\cdot\weight(\bzmi[n])\diff\bzmi[n]
\le
\int\limits_{\support^{n-1}} H'(\bzmi[n])\cdot \weight(\bzmi[n])\diff\bzmi[n] \cdot
\int\limits_{\support^{n-1}} \ind_X(\bzmi[n])\cdot\weight(\bzmi[n])\diff\bzmi[n].
\]
The last inequality for an arbitrary non-negative monotonically increasing function (in our case it is $\ind_X(\cdot)$), an arbitrary non-negative monotonically decreasing function (in our case it is $H'(\cdot)$), and an arbitrary product measure (in our case it is $\weight(\bzmi)$) is known as a special case of Fortuin-Kasteleyn-Ginibre (FKG) inequality~\cite{FortuinKG71} (in the uni-variate case this inequality is called Chebyshev Integral inequality). Therefore, the inequality holds and the claim follows.

%
%
\end{proof}


\section{$k$-Item Vickrey Auction}

In this section, we consider the benchmark $\maxV(\bids)= \max_{1\le k<n} \ k\cdot \bid_{k+1}$ provided by the largest revenue of the $k$-item Vickrey auction across all possible values of supply $k$.

First of all, we calculate the lower bound on the competitive ratio for $\maxV(\cdot)$ using similar approach as in~\cite{GoldbergHKS04}.
We employ equal revenue distribution $\bn$ of bid vectors, that is i.i.d.~with the density function $\weight(b)=\frac{1}{b^2}$, cumulative density $1-\frac{1}{b}$, and the support $[1,\infty)$.
The key technical problem is to compute the expected value of the benchmark $\maxV(\bn)$.
Following~\cite{GoldbergHKS04}, we compute the probability $\prob{\maxV(\bn)\geq z}$ for any given $z$. Since $\maxV(\bn)$ is at least $n-1$, we only need to compute the probability for $z\ge n-1$. Let a random variable $V_i$ be the $i$-th largest bid in $\bn$. We further define a set of random variables as
\[F_{n,k}=\max_{i=1,2,\ldots, n} (k+i-1)\cdot V_i.\]
We note that so far we follow the idea of~\cite{GoldbergHKS04}, but adjust definitions accordingly to the benchmark $\maxV(\cdot)$.
Intuitively, $F_{n,k}$ captures the value of $\maxV(\cdot)$ given $k$ additional bidders with bid $V_1$.
Let ${\cal H}_i$ denote the event
\[V_i \geq \frac{z}{k+i-1} \mbox{ \ and \ } \bigwedge_{j=i+1, i+2, \ldots, n} V_j <   \frac{z}{k+j-1}. \]
The probability of ${\cal H}_i$ can be written as
\[\prob{{\cal H}_i}= {n\choose i} \left(\frac{k+i-1}{z}\right)^i \prob{F_{n-i, k+i}<z}.\]

Since ${\cal H}_i$'s are mutually exclusive and the event $F_{n,k}\geq z$ is the union of  ${\cal H}_i$ for $i=1,2,\cdots, n$, we get
\begin{equation}\label{equ:F-n-k}
\prob{F_{n,k}\geq z} = \sum_i \prob{{\cal H}_i} = \sum_i {n\choose i} \left(\frac{k+i-1}{z}\right)^i \prob{F_{n-i, k+i}<z}.
\end{equation}
This gives a recursive relation for $\prob{F_{n,k}\geq z} $ and the boundary condition is $\prob{F_{0,k}\geq z}=0 $.

We shall prove that
\[\prob{F_{n,k}\geq z}= 1- \frac{(z-n-k+1) (z+1-k)^{n-1}}{z^n} .\]
We verify this by induction. The base case $n=0$ can be verified directly. Further, by \eqref{equ:F-n-k} and by introduction hypothesis, we have
\begin{align*}
\prob{F_{n,k}\geq z} & = \sum_{i=1}^n {n\choose i} \left(\frac{k+i-1}{z}\right)^i \big(1-\prob{F_{n-i, k+i} < z}\big)\\
& = \sum_{i=1}^n {n\choose i} \left(\frac{k+i-1}{z}\right)^i \frac{(z-n-k+1) (z+1-k-i)^{n-i-1}}{z^{n-i}}\\
& = \frac{(z-n-k+1)}{z^n} \sum_{i=1}^n {n\choose i} (k+i-1)^i (z+1-k-i)^{n-i-1}
\end{align*}
By a version of Abel's Identity~\cite{GoldbergHKS04}, we have
\[\sum_{i=0}^n {n\choose i} (k+i-1)^i (z+1-k-i)^{n-i-1}=\frac{z^n}{z-n-k+1}.\]
Substituting this back, we get
\begin{eqnarray*}
\prob{F_{n,k}\geq z} &=& \frac{(z-n-k+1)}{z^n} \left(\frac{z^n}{z-n-k+1} -(z+1-k)^{n-1}\right) \\
&=& 1- \frac{(z-n-k+1)(z+1-k)^{n-1}}{z^n}.
\end{eqnarray*}
This completes the proof of the inductive step.

Therefore, we have
\[\prob{\maxV(\bn)\geq z}=\prob{F_{n,0}\geq z}= 1- \frac{(z-n+1)(z+1)^{n-1}}{z^n}.\]
Now we can get the expectation by the following calculation.
\begin{align*}
\expect{\maxV(\bn)} & =\int_{0}^{\infty} \prob{\maxV(\bn)\geq z} dz \\
&= n-1+  \int_{n-1}^{\infty}\left( 1- \frac{(z-n+1)(z+1)^{n-1}}{z^n}\right) dz\\
&=n-1+  \int_{n-1}^{\infty}\frac{\sum_{i=0}^{n-2} {n \choose i} (n-i-1) z^i }{z^n} dz\\
&=n-1+  \sum_{i=0}^{n-2} {n \choose i} (n-i-1) \int_{n-1}^{\infty}\frac{ 1 }{z^{n-i}} dz\\
&=n-1+  \sum_{i=0}^{n-2} {n \choose i} (n-i-1) \frac{1}{(n-i-1) (n-1)^{n-i-1}}\\
&=n-1+  \frac{1}{(n-1)^{n-1}} \big(n^n- (n-1)^n - n (n-1)^{n-1}\big)\\
&=n \left(\frac{n}{n-1}\right)^{n-1}-n.
\end{align*}
Thus, we get a lower bound on the competitive ratio for $\maxV(\cdot)$
\[\gamma_n \triangleq \frac{\expect{\maxV(\bn)}}{n} = \left(\frac{n}{n-1}\right)^{n-1}-1,\]
where the denominator is $n$, as the expected revenue of any auction is at most $n$ on the equal revenue distribution.
Finally, we note that $\gamma_n$ approaches $e-1$ as $n$ goes to infinity.

\medskip
On the other hand, we can show that $\maxV(\cdot)$ is $\gamma_n$-attainable by our characterization theorem.
The proof is almost identically to our argument for $\F^{(2)}$ in the previous section. We only replace $\lambda_n$ by $\gamma_n$, $\F^{(2)}(\cdot)$ by $\maxV(\cdot)$, and $H_n(\bz)=\max\big(0,n+1-F_n(\bz)\big)$
by $H_n(\bz)=\max\big(0,n-F_n(\bz)\big)$.  In the argument for $\F^{(2)}$, we do not use any explicit formula for $\F^{(2)}$, but employ a few properties such as monotonicity, symmetry, linearity, and the fact that the function value does not change if one increases the largest bid. These properties continue to hold for the benchmark $\maxV(\cdot)$. The definition of $H_n(\bz)=\max\big(0,n+1-F_n(\bz)\big)$ is the only place where we use the formula for $\F^{(2)}$, which is replaced by $H_n(\bz)=\max\big(0,n-F_n(\bz)\big)$ for $\maxV(\cdot)$. Note that the properties that we use are that $H_n(\bz)$ is decreasing and does not change if one increases the largest bid, which continue to hold for the modified $H_n(\bz)$. Therefore, the whole proof of the attainability of $\F^{(2)}$ carries over to the benchmark $\maxV(\cdot)$.

We summarize our results in the following theorem.

\begin{theorem}
For the benchmark $\maxV(\cdot)$ and any $n\ge 2$, the optimal competitive ratio of a truthful auction is $\gamma_n = \big(\frac{n}{n-1}\big)^{n-1}-1$.
\end{theorem}

\section{Concluding Remarks}

Our paper studies designing optimal competitive digital goods auctions.
The proof of the characterization theorem gives an explicit procedure that constructs an auction with the optimal competitive ratio. The construction, however, is rather abstract and may take exponential steps for some benchmarks. It is therefore an intriguing question to see if there is any simple form for the description of optimal auctions (with respect to, e.g., the $\F^{(2)}$ benchmark).

We remark that our characterization applies not only to symmetric benchmarks such as $\F^{(2)}$ but also non-symmetric benchmarks. In particular, Theorem~\ref{th:characterization} may better estimate the competitive ratio of the monotone-price benchmark $\M^{(2)}$ of Leonardi and Roughgarden~\cite{LR2012}.

All our discussions in this paper are only for the unlimited supply case. However, our results to some extent
carry over to the limited supply setting. In a limited supply setting, let $k$ be the number of units for sale and $\nbidder> k$ be the number of bidders.
Let $f: \R^{\nbidder} \rightarrow\R$ be any non-negative and monotone benchmark.
We define $f_1(\bidi[1],\ldots,\bidi[k])= f(\bidi[1],\ldots,\bidi[k-1],\bidi[k],\bidi[k],\ldots,\bidi[k])$ and $f_2(\bidi[1],\ldots,\bidi[k])=f(\bidi[1],\ldots,\bidi[k-1],\bidi[k],0,\ldots,0)$, for each $\bids=(\bidi[1],\ldots,\bidi[n])$ with $\bidi[1]\ge \bidi[2]\ge \cdots \ge \bidi[n]$.
That is, we `increase' and `decrease' in $f(\cdot)$ the last $n-k$ bids to $\bidi[k]$ and 0, respectively.
We have $f_1(b_1,\ldots,b_n)\ge f(b_1,\ldots,b_n) \ge f_2(b_1,\ldots,b_n)$ for any $\bids$ due to the monotonicity condition.
We note that $f_1(\cdot)$ and $f_2(\cdot)$ only depend on the highest $k$ bids and may be viewed as two other benchmarks.

Now we may find the optimal competitive ratios of unlimited supply auctions with $k$ bidders for both benchmarks $f_1(\cdot)$ and $f_2(\cdot)$ with the help of Theorem~\ref{th:characterization}. We note that the competitive ratios for the $f_1(\cdot)$ and $f_2(\cdot)$ benchmarks give upper and lower bounds, respectively, on the competitive ratio of $k$-unit auctions with respect to the benchmark $f(\cdot)$. Indeed, we may construct the following limited supply auction: given a vector of bids $\bids=(\bidi[1],\ldots,\bidi[n])$ with $\bidi[1]\ge \bidi[2]\ge \cdots \ge \bidi[n]$, we admit only the $k$ highest bids and then run the optimal competitive auction with respect to the $f_1(\cdot)$ benchmark. We note that the revenue of such an auction can only be higher than the respective revenue of the unlimited supply auction for $f_1(\cdot)$.

We observe that if $f(\cdot)$ only depends on the $k$ highest bids, then $f(\cdot)=f_1(\cdot)=f_2(\cdot)$ and their respective competitive ratios are the same. This is the case for the benchmark $\F^{(2,k)}$, considered previously in \cite{HartlineM05,GoldbergHKSW06} and defined as the optimal omniscient fixed price auction that sells between 2 and $k$ items. Therefore, without much work we obtain optimal competitive $k$-unit auctions with respect to the benchmark $f(\cdot)$. We note that for other benchmark functions, the question of designing optimal competitive $k$-unit auctions remains open; we leave this interesting question for the future work.

Finally, we remark that the type of questions that we have addressed in Theorem \ref{th:characterization} (namely, what are the worst-case distributions of the input for analysing competitive ratios of various benchmarks) is fundamental to our understanding and discovery of the optimal competitive ratios. We believe that this question should be added to the agenda of the areas in Theoretical Computer Science that use competitive analysis such as online algorithms and algorithmic mechanism design.

\section{Acknowledgments}
We thank Jason Hartline for helpful discussions and valuable comments.

\bibliographystyle{plain}
\bibliography{bibs,game}

\end{document}